\documentclass[12pt,a4paper]{article}
\usepackage[utf8]{inputenc}
\usepackage{amsmath} %
\usepackage{amsfonts} %
\usepackage{amssymb} %
\usepackage{graphicx} %
\usepackage{xcolor}
\usepackage{comment}
\usepackage{amsthm}
\usepackage[shortlabels]{enumitem}
\usepackage{mathrsfs}
\usepackage{hyperref} %
\usepackage{chngcntr}
\counterwithin{figure}{section}

\usepackage{definitions}
\newtheorem{theorem}{Theorem}

\newtheorem{lemma}{Lemma}
\theoremstyle{remark}

\newtheorem{definition}{Definition}

\newcommand\blfootnote[1]{%
  \begingroup
  \renewcommand\thefootnote{}\footnote{#1}%
  \addtocounter{footnote}{-1}%
  \endgroup
}

\newif\ifR %
\Rtrue 

\title{\bf {Optimal distributed composite testing in high-dimensional Gaussian models with 1-bit communication}}

 \author{Botond Szab\'o,\thanks{Department of Decision Sciences, Bocconi University, Bocconi Institute for Data Science and Analytics (BIDSA), \href{mailto:botond.szabo@unibocconi.it}{botond.szabo@unibocconi.it}} \and Lasse Vuursteen\thanks{Delft Institute of Applied Mathematics (DIAM), Delft University of Technology, \href{mailto:l.vuursteen@tudelft.nl}{l.vuursteen@tudelft.nl}} \and and Harry van Zanten\thanks{Department of Mathematics, Vrije Universiteit Amsterdam, \href{mailto:j.h.van.zanten@vu.nl}{j.h.van.zanten@vu.nl}}  } 

\date{}

\providecommand{\keywords}[1]
{
  \small	
  \textbf{\textit{Index terms---}} #1
}

\begin{document}

\maketitle

\begin{abstract}
In this paper we study the problem of signal detection in Gaussian noise
in a distributed setting where the local machines in the star topology can communicate a single bit of information. We derive a lower bound on the {Euclidian norm} that the 
signal needs to have in order to be detectable. Moreover, we 
exhibit  optimal distributed testing strategies that attain the lower bound.
\end{abstract}



\section{Introduction}

\blfootnote{2022 IEEE. Personal use of this material is permitted.
  Permission from IEEE must be obtained for all other uses, in any current or future 
  media, including reprinting/republishing this material for advertising or promotional 
  purposes, creating new collective works, for resale or redistribution to servers or 
  lists, or reuse of any copyrighted component of this work in other works. 
  DOI: \href{https://doi.org/10.1109/TIT.2022.3150599}{10.1109/TIT.2022.3150599}}

The rapidly increasing amount of available data in many fields of application has triggered the development of distributed methods
for data analysis.  
Distributed methods, besides being able to speed up  computations considerably, 
can reduce local memory requirements  and  can also help in protecting  privacy, by 
refraining from storing a whole dataset in a single central location. 
Moreover, distributed methods occur naturally when data is by construction observed and processed 
at multiple locations, for instance in astronomy, meteorology, seismology, military radar or air traffic control systems.

{In the context  of decentralized detection, the task of distinguishing between different signals based on 
 information provided from a network of sensors has been investigated since the late seventies. 
 Motivated by applications such as surveillance systems and wireless communication, statistical hypothesis testing with distributed sensors has received considerable attention over the past few decades, e.g.\ \cite{tenney_detection_1981,tsitsiklis_decentralized_1988,kreidl_decentralized_2011-1,tarighati_decentralized_2017}. Most of this literature has been concerned with distinguishing 
 between finitely many signals, by combining the decisions of sensors or machines receiving (noisy) data from the same underlying signal \cite{varshney_distributed_1997, chamberland_wireless_2007}. The statistical hypotheses considered at the time where either simple, or reduce to a simple hypothesis. Under the name of multiterminal data compression, hypothesis testing and estimation with observations in fixed finite sample spaces (alphabets) was investigated \cite{ahlswede_hypothesis_1986,han_hypothesis_1987,han_exponential-type_1989,amari_fisher_1989-1,shalaby_multiterminal_1992,watanabe_neymanpearson_2018}. In this body of literature, each terminal observes a stream of iid observations in some alphabet, where the different terminals might receive observations from different distributions on different alphabets. Each of the terminals compresses the iid observations into a message that is sent to a central terminal, at which some inference goal is to be achieved based on the messages. For an overview, see \cite{te_sun_han_statistical_1998}.

The information theoretic aspects of distributed statistical methods have only
been studied rigorously in terms of sample complexity relatively recently. The problem of distributed testing with a composite alternative hypothesis about a high-dimensional parameter with limited communication was first 
considered in a few recent papers. In \cite{acharya:etal:2020}, theoretical guarantees are derived for distributed uniformity testing of a discrete distribution in the case that a collection of  machines each receive one observation. Here, sample complexity (in the sense of minimax convergence rate) is assessed in terms the cardinality of the sample space, the number of bits available for communication and the number of observations. In that paper, it is also shown that in the distributed setting with many local machines, 
testing performance can strictly improve when all machines have access to a shared source of 
randomness, a so-called {\em public coin}.
Distributed uniformity testing of a discrete distribution when multiple observations per machine are available is considered in \cite{pmlr-v99-diakonikolas19a, fischer:etal:2018}. 

Most of the work, studying distributed inference in terms of sample complexity, up until now has focused on 
distributed methods for estimating a signal in the normal-means model under bandwidth, or communication restrictions (see 
for instance \cite{zhang:2013, braverman:etal:2016,xu_information-theoretic_2017,barnes2020lower,cai_distributed_2020}). The canonical normal-means model postulates that we have an 
observation $X$ satisfying
\begin{align*}
X=\mu+\frac{1}{\sqrt{n}}Z,
\end{align*}
where $\mu\in\mathbb{R}^d$ is the unknown signal, $Z\sim N_d(0,I_d)$ is an unobserved noise vector
with a $d$-dimensional standard Gaussian distribution, and $n$ is the signal-to-noise ratio.
Note that this is equivalent to observing $n$ independent copies of a $N_d(\mu, I_d)$-vector. In the distributed setting considered in the aforementioned articles, $n$ such observations are distributed across $m$ machines, which then communicate a transcript to a central machine. The central machine then forms an estimate of $\mu$ on the basis of these transcripts. The sample complexity of such a problem is expressed in $n$, $m$, the number of bits available for communication of the transcripts and the properties of $\mu$ (e.g. Euclidian norm). Related to this normal-means model, on deriving minimax lower bounds and optimal distributed estimation strategies 
 in the context of nonparametric regression, density estimation and Gaussian signal-in-white-noise models (e.g.
\cite{szabo:zanten:2018,zhu:2018, barnes2020lower,han2018geometric, szabo:zanten:2020}).}

In this paper, we investigate the information theoretic properties of distributed methods for 
{\em testing} for the presence of a signal in the normal-means model. 
The theory on distributed testing in this setting is much less developed than that for estimation. Testing for the presence of a signal in the normal-means model translates to testing the null hypothesis $H_0: \mu = 0$
that the sequence is identically equal to $0$. 
Rejecting this hypothesis means declaring that there is a non-zero signal. 
A fundamental question is how
large the signal should be in order to be detectable. 
It is well known that in this non-distributed model, the size of the signal  ({by which we mean the Euclidean norm}) has to be of order $d^{1/4}/\sqrt{n}$ in order for the signal to be 
detectable (see e.g.\ \cite{baraud:2002}). 
An optimal test is for instance obtained by rejecting the hypothesis $H_0$
if $\|X\|^2$ is above a specific threshold, depending on $d$ and $n$ (see also Section \ref{sec: upper}).

The question we address in this paper is how this changes in the distributed setting. 
In our analysis we consider the distributed version of the normal-means model, 
 where the data ($n$ independent draws from the $N_d(\mu,I_d)$-distribution) is divided over $m \le n$ 
 local machines, or cores. Equivalently, we assume that at 
 each local machine $j\in\{1,...,m\}$ we observe a vector $X^{(j)}$ which satisfies
\begin{equation}\label{dynamics_Xj}
X^{(j)} = \mu + \sqrt{\frac{m}{n}} Z^{(j)},
\end{equation}
where again $\mu \in \mathbb{R}^d$ and the $Z^{(j)}$ are independent  $N_d(0,I_d)$-distributed vectors.
Each machine carries out a test for the hypothesis $H_0: \mu = 0$ using its local data $X^{(j)}$, 
where we allow that the machines use a public coin, i.e.\ a common random vector $U$ that is available to all local machines.
Subsequently, the outcomes of the $m$ local tests (which are single bits) are sent to a central machine, where 
they are combined into a single, overall test.
We prove that in this distributed setting, the size of the signal   
has to be of the order $(d (m \wedge d))^{1/4}/\sqrt{n}$ in order for the signal to be 
detectable. Moreover, we exhibit optimal distributed tests that achieve this detection bound.

The detection bound has a remarkable ``{regime change}'' or ``{elbow effect}'' at $m = d$. 
As $m$ grows from $1$ to $d$ the testing problem becomes more difficult, in the sense that 
the signal needs to be ever larger to be detectable using an increasing number of machines.
This is intuitively understandable, since as $m$ grows, the local signal-to-noise ratio 
decreases, so it is reasonable to expect that the signal needs to be larger to be able to detect it.
The detection bound stops increasing if $m$ grows above $d$ however. 
In that range, the decrease of the local signal-to-noise ratios is apparently balanced 
by  the increase of the number of bits that are transmitted from the local machines to the 
central one.

The {regime change} is also reflected by the fact that we need different testing strategies 
depending on how $m$ and $d$ are related. If $m$ is below  some threshold, and in particular
it does not increase with $n$, then, not surprisingly, we can simply  use the 
classical non-distributed test mentioned above at one of the local machines. If $m$ is larger than this 
threshold but  $m \le d$, then it is still possible to construct an optimal test 
by combining local tests that are based on the test statistics $\|X^{(j)}\|^2$, but the 
test needs to be constructed more carefully.
In the range  $m \ge d$ this strategy becomes sub-optimal and we have to 
adopt a different approach, using the fact 
that we have  a public coin at our disposal. For this case
we construct an optimal distributed test by appropriately combining local tests that use 
local test statistics of the form $U^\top X^{(j)}$, where $U$ is a public random vector. 

{ The approach to finding the lower bound can be sumarized as follows. As a first step, we lower bound the testing risk by a type of Bayes risk, where $\mu$ is drawn from a prior such that it either belongs to the null hypothesis of the alternative, as in for example \cite{ingster:2002}. This Bayes risk can be related to the mutual information between the testing outcome and which hypothesis is selected, akin to techniques common in tackling (distributed) estimation problems through Fano-like  inequalities in for example \cite{zhang:2013, chen_on_bayes_risk_lower_bounds}. In particular, the tensorization property of the mutual information is used and combined with a so called strong data processing inequality to quantify the loss incurred in the distributed setup, similar to the approaches to distributed estimation in \cite{xu_information-theoretic_2017,braverman:etal:2016,cai_distributed_2020}.}

Upon completion of this work, we came across the paper \cite{pmlr-v125-acharya20b}, which considers a 
setting similar to ours and claims some partly overlapping results. 
There are also  important differences between the papers, however. { Most importantly perhaps, our proof strategy for the lower 
bound is rather different. As a result our proof is arguably easier to verify and 
at least provides an alternative to the approach proposed in \cite{pmlr-v125-acharya20b}. Furthermore, the paper \cite{pmlr-v125-acharya20b} does not allow the number 
of machines (our $m$) to vary. Essentially,  only the case that $m=n$ is considered. }

The remainder of the paper is organized as follows. In Section \ref{sec: setting} we formally describe the model and the distributed testing problem, and introduce notations used throughout the paper. In Section \ref{sec: main} we present our main results. We state the detection  lower bound in Section \ref{sec: lower}, and we provide novel distributed tests achieving the theoretical limits in Section \ref{sec: upper}. We provide a short simulation study demonstrating the {regime change} observed in the theoretical analysis in Section \ref{sec: simulation}. 
The proofs for the distributed tests achieving the minimax rate is given in Section \ref{sec: proof:thm:UB}, while the proofs of corresponding (technical) lemmas are deferred to Sections \ref{sec: thm_upper_bound}-\ref{sec: technical:lemmas}.

\section{Problem setting and notation}\label{sec: setting}

We assume we have $m$ local machines. For $j = 1, \ldots, m$, we have an observation $X^{(j)}$ 
at machine $j$, which satisfies 
\begin{equation*}
X^{(j)} = \mu + \sqrt{\frac{m}{n}} Z^{(j)}. 
\end{equation*}
Here $\mu \in \mathbb{R}^d$ is the unknown signal of interest and $Z^{(1)}, \ldots, Z^{(m)}$ 
are independent, $N_d(0,I_d)$-distributed vectors.
 We allow both the dimension $d=d_n$ and the number of machines $m=m_n$ to depend on 
the overal signal-to-noise ratio, or ``sample size''
$n$. In fact, the interesting cases (from an asymptotic perspective) 
are the ones where both $d$ and $m$ are tending to infinity with $n$. 
Nevertheless, we do not restrict ourselves to this asymptotic regime 
and cover the finite $m$ and $d$ cases as well.

We are interested in distributed tests for the hypotheses  
\begin{equation}\label{eq: hyp}
H_0: \mu = 0, \qquad \text{against} \qquad H_{\rho} : \|\mu \| \geq \rho,
\end{equation}
for $\rho > 0$. 
We consider  public coin protocols, where each machine has access to a shared random vector $U$, which is independent of the observations $X^{(1)},\ldots,X^{(m)}$. 
Each local machine $j$  carries out a local test. Using the local data $X^{(j)}$ and the public coin $U$ 
it produces a  binary,  $\{0,1\}$-valued outcome $T^{(j)}$. 
The outcomes $T^{(1)}, \ldots, T^{(m)}$ of the local tests are transmitted to a 
central machine where they are aggregated into a global test, described by a $\{0,1\}$-valued
variable $T_{dist}$.  Schematically, a distributed test looks as follows:
\begin{equation} \label{def: dist:test}
 \begin{matrix}
(X^{(1)},U) &\put(0,5){\vector(1,0){20}}  &\qquad T^{(1)}\quad&\put(-10,4){\vector(3,-1){20}} \\
\vdots &\put(0,5){\vector(1,0){20}}  &\qquad\vdots  \quad&\put(-10,5){\vector(1,0){20}} \\
(X^{(m)},U) & \put(0,5){\vector(1,0){20}}  &\qquad T^{(m)}\quad&\put(-10,5){\vector(3,1){20}}
\end{matrix}  \qquad
 T_{dist}.
\end{equation}
We denote the collection of all  distributed tests of this form by $\mathcal{T}_{dist}$.

The testing error, or risk of a distributed test $T_{dist}$, is defined as usual by 
\begin{equation}
\mathcal{R}(H_{\rho},T_{dist}) = P_0\left( T_{dist} = 1 \right) + \underset{\|\mu\| \ge \rho}{\sup} P_\mu\left( T_{dist} = 0 \right)\label{def: test:error},
\end{equation}
i.e.\ as the sum of the type one and type two errors of the test. (Here, and elsewhere,  we denote 
by $P_\mu$ the underlying distribution assuming that $\mu$ is the true signal.) 
Uniform lower bounds for this risk express the impossibility of 
detecting a signal of size $\rho$. Indeed, 
 fix a level $\alpha \in (0,1)$. If $\rho > 0$ is such that 
$\mathcal{R}(H_{\rho},T_{dist}) > \alpha$ for all $T_{dist} \in \mathcal T_{dist}$, 
then it means that there exists no consistent level-$\alpha$ test for testing $H_0$
against $H_\rho$. In other words, no distributed test of level $\alpha$ is able 
to detect all signals that are larger than $\rho$ in Euclidean norm. 

Our aim is to find the detection threshold, i.e.\ the cut-off $\rho_{dist}$ 
such that no consistent level-$\alpha$ test exists if $\rho < \rho_{dist}$
and at least one consistent level-$\alpha$ test exists if $\rho \ge \rho_{dist}$. 
We will show that, up to constants depending on the chosen level $\alpha$, 
the detection threshold is given by 
\[
\rho^2_{dist} \asymp \min\Big\{\frac{\sqrt{dm}}{n}, \frac{d^{}}{{n}}\Big\}.
\]
Moreover, we exhibit optimal tests for the case $\rho \ge \rho_{dist}$.

\subsection{Notation}
We write $a \wedge b = \min\{a, b\}$ and $a \vee b = \max\{a, b\}$. 
For two positive sequences $a_n,b_n$ we use the notation $a_n\lesssim b_n$ if there exists a universal positive constant $C$ such that $a_n\leq C b_n$. We write $a_n\asymp b_n$ which holds if $a_n\lesssim b_n$ and $b_n\lesssim a_n$ are satisfied simultaneously.  
The Euclidean norm of a vector $v \in \mathbb{R}^d$ is denoted by $\|\cdot\|$. 
For absolutely continuous probability measures $P\ll Q$, we denote by $D_{KL}(P\| Q)=\int \log\frac{dP}{dQ}dP$ their Kullback-Leibler divergence.

\section{Main results}\label{sec: main}

%

\subsection{Lower bound for the detection threshold}\label{sec: lower} 

The following theorem establishes the detection threshold. 
Its proof is described in the remainder of the subsection.

\begin{theorem}\label{lower_bound_thm}
Fix $\alpha \in (0,1)$ and suppose that
\begin{equation}\label{lower_bound_rate}
\rho^2 < c_{\alpha} \frac{\sqrt{d (m\wedge d)}}{n}
\end{equation}
for ${c_{\alpha} \leq ({1-\alpha})^2/384}$. Then,
\begin{equation*}
\underset{T\in\mathcal{T}_{dist}}{\inf}\; \mathcal{R}(H_{\rho},T) > \alpha,
\end{equation*}
where infimum is  over all distributed tests $T\in\mathcal{T}_{dist}$ given in \eqref{def: dist:test}. 
\end{theorem} 

\bigskip

The result tells us that if \eqref{lower_bound_rate} holds, there does not exist a  consistent test in $\mathcal{T}_{dist}$ of level $\alpha \in (0,1)$ for the hypotheses \eqref{eq: hyp}. In other words, no distributed test can detect all signals 
of size $\rho$. { It should be noted that we did not optimize for the value of the constant $c_\alpha$ and the statement is likely to be still true for larger values of $c_\alpha$.}

The proof of the theorem relies on three key lemmas, which we state below after introducing some necessary notations. As a first step, we use the basic fact that the supremum of the probability of a type two error of a test can be lower bounded by a Bayesian type two 
error, i.e.\ for any prior distribution $\Pi$ supported on $H_\rho$
\begin{equation*}
\underset{\mu \in H_{\rho}}{\sup} P_\mu\left( T = 0 \right) \geq \int_{H_{\rho}} P_\mu\left( T = 0 \right)\,d\Pi(\mu).
\end{equation*}

To further lower bound the risk we construct an appropriate Markov chain
and relate the testing problem to an information transfer problem through the chain.
Consider $V\sim \text{Ber}(1/2)$, i.e.\ a $V$ is $0$ or $1$, each with probability $1/2$, independent of the public coin random 
vector $U$, such that the random vectors $X^{(j)}|(V=0)$, $j=1,\dots,m$ follow \eqref{dynamics_Xj} with $\mu = 0$ and $X^{(j)}|(V=1)$ follows a Gaussian mixture $P_{\Pi}$ defined as $P_{\Pi}(A) =\int P_\mu(A)\, d\Pi(\mu)$ for all Borel sets $A\subset\mathbb{R}^{d}$. Let us denote by $\P$ the joint probability measure describing the corresponding Markov dynamics
\begin{equation} \label{def: dist:test:MC}
V  \quad {{\vector(1,0){20}} \quad \mu} \qquad 
 \begin{matrix}
\put(-10,-4){\vector(3,1){20}} &\quad(X^{(1)},U) &\put(0,5){\vector(1,0){20}}  &\qquad T^{(1)}\quad&\put(-10,4){\vector(3,-1){20}} \\
\put(-10,3){\vector(1,0){20}}&\quad\vdots &\put(0,4){\vector(1,0){20}}  &\qquad\vdots  \quad&\put(-10,4){\vector(1,0){20}} \\
\put(-10,12){\vector(3,-1){20}}&\quad(X^{(m)},U) & \put(0,5){\vector(1,0){20}}  &\qquad T^{(m)}\quad&\put(-10,5){\vector(3,1){20}}
\end{matrix}  \qquad
 T.
\end{equation}

We then have that for any distributed test $T$,
\begin{equation}\label{starting_pt_display_proof_mut_info_lem}
\mathcal{R}_{}(H_{\rho},T)\geq \P(T=1|V=0)+\P(T=0|V=1)  = 2\P(T \neq V). 
\end{equation}
The right hand side of \eqref{starting_pt_display_proof_mut_info_lem} can be further  bounded from below using the  \emph{mutual information} between $T$ and $V$ in the chain \eqref{def: dist:test:MC}, defined by
\begin{equation*}
I_\Pi (V,T) = D_{KL}\left( \P^{V \times T} \, \| \; \P^V \times \P^T \right),
\end{equation*}
where $\P^V$, $\P^T$ and $\P^{V \times T}$ denote marginal- and joint distributions of $V$ and $T$, and the subscript $\Pi$ is used to indicate the dependence on the prior $\Pi$. Informally, the mutual information measures how much knowing $T$ reduces uncertainty about $V$ and vice versa. The lower bound based on the mutual information is given in the following lemma. 
The proof of the lemma is deferred to Section \ref{proof: lem1}.

\bigskip

\begin{lemma}\label{lemma_mutual_information_testing_lb}
Let $\Pi$ be a prior on $H_\rho$ and consider the dynamics \eqref{def: dist:test:MC}. We have
\begin{equation*}
\underset{T \in \mathcal{T}_{dist}}{\inf}\; \mathcal{R}_{}(H_{\rho},T)  \geq   1 -  \sqrt{2I_\Pi(V,T)}.
\end{equation*}
\end{lemma}

\bigskip

In view of the usual data processing inequality we have $I_\Pi(V,T) \leq I_\Pi(V,(T^{(1)},\dots,T^{(m)}))$. 
The following lemma asserts that, { up to an additional term,} this further tensorizes conditional on the public coin randomness. 

\bigskip

\begin{lemma}\label{lem: tensor}
Consider the dynamics \eqref{def: dist:test:MC}. We have 
\begin{equation}\label{eq : tensored mutual information}
I_\Pi(V,(T^{(1)},\dots,T^{(m)})) \leq  \underset{j=1}{\overset{m}{\sum}} I_\Pi(V,T^{(j)} |U) + { \underset{j=1}{\overset{m}{\sum}} I_\Pi(\mu, T^{(j)} |U,V)}. 
\end{equation}
\end{lemma}

\bigskip

The proof of this lemma is given in Section \ref{proof: lem_tensor}. This bound, combined with Lemma \ref{lemma_mutual_information_testing_lb}, allows us to break down the difficulty of the 'global' testing problem in terms of the difficult of the $m$ 'local' testing problems, captured by the quantities $I_\Pi(V,T^{(j)} |U)$. These conditional local mutual informations quantify the capacity of the local tests to distinguish a signal drawn from the prior $\Pi$ from 
the zero signal. {The second sum in the display of the lemma captures dependency between the transcripts and the prior draw $\mu \sim \Pi$. Essentially, it captures how well the signal can be \emph{estimated} by the local tests. }

We now discuss the choice of prior distribution $\Pi$. Let us set $\epsilon=\rho/\sqrt{d}$, let $R$ be a $d$-dimensional vector of independent Rademacher random variables, and define the prior $\Pi$ as the distribution of $\epsilon R$. 
Note that $\Pi$ has support contained in $H_\rho$. (Such choices are typically considered as least favorable priors supported on signals that are 
difficult to detect, see for instance Section 3.2 of \cite{ingster:2002}.)

Since $V$, $\mu$ and $X^{(j)}$ are independent of $U$, conditioning on $U$ does not distrupt the Markov chain property: 
we have the chain $V|U \to \mu | U \to X^{(j)}|U \to T^{(j)}|U$. { Consequently, the ``estimation term'' $I_\Pi(\mu, T^{(j)} |U,V)$ can be handled using strong data processing techniques employed in distributed estimation, see for example Lemma 11 in \cite{cai_distributed_2020}. For completeness, we adopted the aforementioned lemma for the above choice of prior distribution in the form of Lemma \ref{lem : multivar gaussian estimation SPDI} in the appendix, which yields that $I_\Pi(\mu, T^{(j)} |U,V) \leq 128 \frac{n \rho^2}{d m} I_\Pi(X^{(j)} , T^{(j)} |U,V=1)$. Using that $T^{(j)}$ is binary valued, we obtain that the second term in \eqref{eq : tensored mutual information} is bounded above by $128 \frac{n \rho^2}{d}$.}  

The loss of information about $V$ resulting from the compression of $X^{(j)}|U$ into $T^{(j)}|U$ in this Markov chain is quantified by Lemma \ref{SDPI} below. {The lemma comes in the form of a strong data processing inequality for the testing problem and forms the crux of the proof of the lower bound. It captures the difficulty of the local testing problem in terms of $n$, $m$, $d$ and $\rho$.}


\bigskip

\begin{lemma}[Public Coin Strong Data Processing Inequality]\label{SDPI}
With $\Pi$ as defined above we have
\begin{equation*}
I_\Pi(V,T^{(j)}|U) \leq (48\beta \wedge 1) I_\Pi(X^{(j)},T^{(j)}|U),
\end{equation*}
where
\begin{equation}\label{beta}
\beta = \begin{cases}
\frac{n^2\rho^4}{dm^2} \;\;\; \text{ if } \frac{m}{n \rho^2} < 1/2, \\
\frac{2n\rho^2}{dm} \;\;\; \text{ if } \frac{m}{n \rho^2} \geq 1/2.
\end{cases}
\end{equation}
\end{lemma}

\bigskip

We give the proof of the lemma in Section \ref{proof: PC-SDPI}. By combining the information theoretic inequalities above 
with the fact that $I(X^{(j)},T^{(j)}|U)\leq H(T^{(j)}|U) \leq 1$, which is true because $T^{(j)}$ is a binary variable,  we get that 
$$I_{\Pi}(V,T)\leq \sum_{j=1}^{m} I_\Pi(V,T^{(j)} |U) + { \underset{j=1}{\overset{m}{\sum}} I_\Pi(\mu, T^{(j)} |U,V) \leq  48\beta m + 128 \frac{n \rho^2}{d}}.$$
Therefore, in view of Lemma \ref{lemma_mutual_information_testing_lb},
\begin{equation*}
 \mathcal{R}_{}(H_{\rho},T)  \geq 1 - 4\sqrt{ 6 \frac{n\rho^2}{d} { \left( \max \left \{ \frac{n \rho^2}{m} , 2 \right \} + 8/3 \right)} }.
\end{equation*}
For $\rho$ satisfying \eqref{lower_bound_rate}, the right-hand side is bounded from below by $\alpha$ for an 
arbitrary distributed test $T\in\mathcal{T}_{dist}$.

\subsection{Optimal tests attaining the lower bound}\label{sec: upper}

In this section, we exhibit a distributed testing procedure that is optimal in the sense 
that it attains the lower bound of Theorem \ref{lower_bound_thm}. More precisely, 
we show that if $\rho^2$ is larger than a multiple of the right-hand side of \eqref{lower_bound_rate}, 
there exists a distributed test for $H_0$ against $H_\rho$ 
with risk bounded by a chosen level $\alpha \in (0,1)$. 
Summarising, we have the following theorem, complementing the lower bound of Theorem \ref{lower_bound_thm}. 

\bigskip

\begin{theorem}\label{thm_upper_bound}
Fix $\alpha \in (0,1)$ and suppose that 
\begin{equation}\label{rho_seq_bound_that_is_uniformly_testable}
\rho^2 \geq C_{\alpha} \frac{\sqrt{d (m\wedge d)}}{n}
\end{equation}
for $C_\alpha > 0$ a constant depending only on $\alpha$ as given in \eqref{def:M_alpha:C_alpha}.
Then there exists a test $T \in \mathcal{T}_{dist}$ such that 
 $ \mathcal{R}_{}(H_{\rho},T) \leq \alpha$.
\end{theorem}

\bigskip

We prove the theorem by constructing three concrete distributed tests, for three different ranges of $m$. We outline the construction of these tests in this section. The detailed verification that they are consistent at the level $\alpha$ for testing $H_0$ against $H_{\rho}$ in their respective ranges  is deferred to Section \ref{sec: proof:thm:UB}.

First assume that the number of machines is large enough, but does not exceed the dimension $d$, i.e. $M_{\alpha}\leq m \leq d$, for a large enough constant $M_{\alpha}>0$ given in \eqref{def:M_alpha:C_alpha}. 
In this case, Theorem  \ref{lower_bound_thm} asserts that the  detection lower bound for $\rho$ is a multiple of $(dm)^{1/4}/n^{1/2}$.
We propose the following distributed test that is able to detect signals with Euclidian norm of that order
if $M_{\alpha}\leq m \leq d$. In this case no public coin is needed.
In this setting we first compute the local test statistic $S_{m \leq d}^{(j)} = (n/m)\|X^{(j)}\|^2$
at every machine $j$. Under the null hypothesis, $S_{m \leq d}^{(j)}$ follows a chi-square distribution with $d$ degrees of freedom, i.e. $S_{m \leq d}^{(j)}\sim \chi^2_d$. Then for every $j$ we consider the randomized test $T_{m \leq d}^{(j)}$ using Bernoulli random variables
\begin{equation*}
T_{m \leq d}^{(j)} | S_{m \leq d}^{(j)} \sim \text{Ber}\left(F_{\chi^2_d} \left(S_{m \leq d}^{(j)} \right)\right), 
\end{equation*}
where $F_{\chi^2_d}$ denotes the distribution function of $\chi^2_d$. Under the null hypothesis the $T_{m \leq d}^{(j)}$
are independent and $\text{Ber}(1/2)$. At the central machines we combine the local test in 
a global test $T_{m \leq d} \in \mathcal{T}_{dist}$ by setting
\begin{equation}
T_{m \leq d} = \mathbbm{1} \left \{ \left| \underset{j=1}{\overset{m}{\sum}} (T_{m \leq d}^{(j)}- 1/2) \right| \geq  \sqrt{m}\bar{\kappa}_\alpha  \right \},\label{def_global_test_small_m}
\end{equation}
with $\bar{\kappa}_\alpha^2 = 3\log \left(4/\alpha \right)$. By a standard Chernoff bound, the type 
one error of this test is bounded by $\alpha/2$
for $m$ large enough. In Section \ref{sec: case1}, we prove that the same is true for the type two error if 
$\|\rho\|^2 \ge C_{\alpha}\sqrt{dm}/n$.  as posed by Theorem 

Next we assume that $m \geq d \vee M_{\alpha}$, in which case the detection lower bound for $\rho$ is a multiple of 
$\sqrt{d/n}$. { When $m/d \to \infty$, tests based on the statistics $S_{m \leq d}^{(j)}$ cannot reach the detection lower bound anymore as can be observed by inspection of the variance of \eqref{Tj_under_Pf}, or through Theorem 16 in \cite{pmlr-v125-acharya20b}.} We propose a novel distributed test using a public coin. 
Specifically, we assume all machine have access to a vector 
 $U=(U_1,\ldots,U_d)$ of independent standard normal random variables. 
 For $j = 1, \ldots, m$ we  compute the local test statistics $S_{m \geq d}^{(j)} = \sqrt{\frac{n}{md}} U^T X^{(j)}$ and the corresponding local tests
\begin{equation*}
 T_{m \geq d}^{(j)} = \mathbbm{1} \left \{ S_{m \geq d}^{(j)} \geq 0 \right \}.
\end{equation*}
Then we aggregate these local tests in the central machine to a distributed test $T_{m \geq d} \in \mathcal{T}_{dist}$
by defining
\begin{equation}
T_{m \geq d} = \mathbbm{1} \left \{ \left| \underset{j=1}{\overset{m}{\sum}} (T_{m \geq d}^{(j)} - 1/2) \right| \geq  \sqrt{m}\tilde{\kappa}_\alpha  \right \},\label{def_global_test_large_m}
\end{equation}
with $\tilde{\kappa}_\alpha^2 = (1/3) \log \left(16/\alpha \right)$. 
In Section \ref{sec: case2} we prove that this test satisfies the required error bound if 
$\rho^2 \ge C_{\alpha}d/n$.

Finally, we consider the case $m\leq M_\alpha$ for completeness. We have to treat it separately for technical reasons, 
although in practice we would probably simply use the first test above for all cases $m \le d$. 
 To achieve optimality in this case we can simply choose  a single machine 
and conduct the hypothesis test we would use  in the classical, non-distributed  setting.
Specifically, we can for instance use as global test $T_{m \asymp 1} \in \mathcal{T}_{dist}$
the test given by 
\begin{equation}
T_{m \asymp 1}  = \mathbbm{1} \left\{  \frac{n}{\sqrt{d}m}\|X^{(1)}\|^2-\sqrt{d} \geq \kappa_{\alpha} \right\},\label{def_global_test_m_is_1}
\end{equation}
with $\kappa_{\alpha}=2/\sqrt{\alpha}$. See Section \ref{sec: case3} for details.

\section{Simulation experiments}\label{sec: simulation}

In this section we investigate the performance of the tests $T_{m \le d}$ and $T_{m \ge d}$
that were proposed in the preceding section, in particular with regards to the 
 ``elbow effect'' when $m \simeq d$ that we observed in our theoretical results.
We visualize the performance using two different simulation experiments. 

In the first experiment
we consider two fixed choices of $m, n$ and $d$, one with $m < d$ and one with $m > d$. 
We simulate  data with a signal with increasing strength $\rho$ and 
then assess the performance of the tests by computing the ``true positive rate'' (TPR), i.e.\ the fraction of the simulations in which they correctly detect the signal.

For the second experiment we also consider two scenarios, one with $m \lesssim d$ and one with $m \gtrsim d$.
But now we fix the signal strength $\rho$ a little above the detection limit in both cases and investigate how
the performance of the tests depends on the total signal-to-noise ratio, 
or sample size $n$.

\subsection{First experiment} 

\begin{figure}[h]\label{fig: increasing alt norm}
\begin{center}
\includegraphics[width=\textwidth]{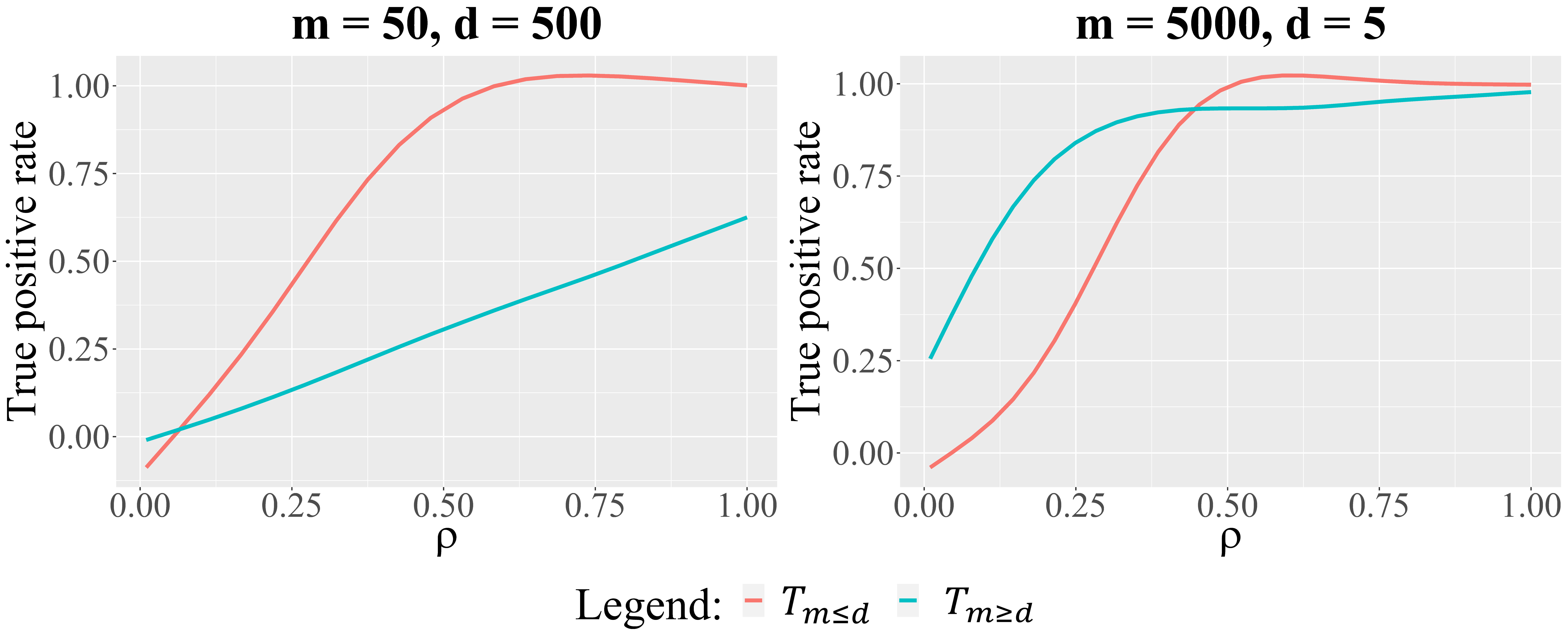}
\end{center}
\caption{The horizontal axes correspond to the Euclidian norm $\|\mu\|=\rho$ of the signal
\eqref{two_point_prior_draws} used in the simulations. 
In the plot  on the left we have $m=50$ and $d=500$, on the right we have $m=5000$ and $d=5$. 
For each $\rho$ in a grid between $0$ and $1$, $100$ datasets were simulated. 
The solid lines give the TPR, i.e.\ the fraction of the  $100$ runs
in which the tests correctly detected the signal. The red lines correspond to $T_{m \le d}$, 
the blue lines to $T_{m \ge d}$.} 
\end{figure}

In the first simulation we consider fixed values for $n$, $m$ and $d$
and we simulate data in which we let the strength $\|\mu\|$ of the unobserved signal vary between $0$ and $1$. 
We investigate how well each of the distributed tests manages to correctly 
reject the null hypothesis, i.e.\ detect the signal.

In Figure \ref{fig: increasing alt norm}, we consider two different scenarios. In both scenarios we choose 
$n=10^4$ and we have specified the rejection criterion for both tests such that they have a type one error probability 
of less than  $\alpha = 0.05$. 
The signal $\mu$ is drawn according to 
\begin{equation}\label{two_point_prior_draws}
\mu = \frac{\rho}{\sqrt{d}} R,
\end{equation}
where $R$ is a vector of independent Rademacher random variables, and we let 
the signal strength $\rho = \|\mu\|$ vary from $0$ to $1$. 

In the $m < d$ scenario corresponding to the plot on the left-hand side in Figure \ref{fig: increasing alt norm}, 
we see that the test $T_{m \le d}$
outperforms the public coin test $T_{m \ge d}$,
 in accordance with our theoretical results.
In the $m > d$ scenario on the right-hand side we see that the test 
$T_{m \ge d}$ detects the presence of the signal much earlier than the test
$T_{m \le d}$.

\subsection{Second experiment} 

\begin{figure}[h]\label{fig: n m d seq}
\includegraphics[width=\textwidth]{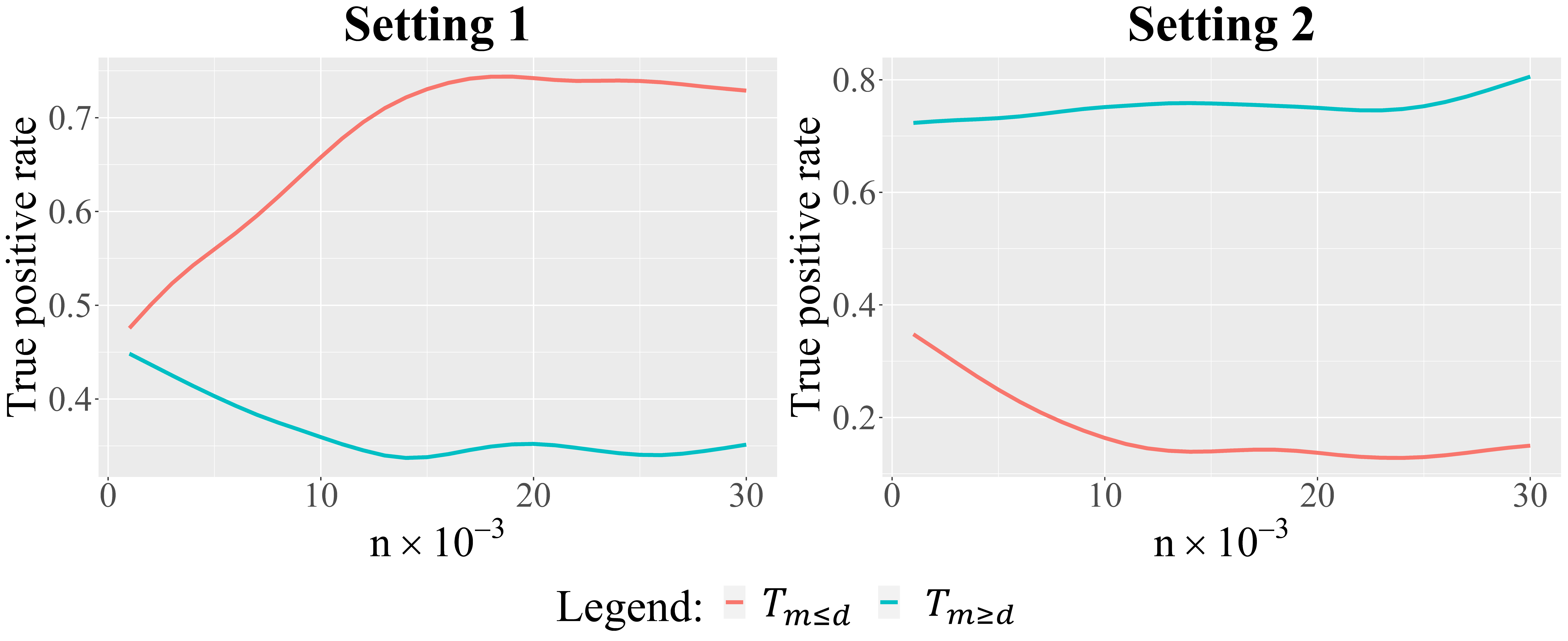}
\caption{The horizontal axes corresponds to the total sample size $n$.
In the plot on the left we have $d= n^{{2}/{3}}$,  $m = 500$,  
and $\rho =  {\log(d) (d m)^{1/4} }/{\sqrt{n}}$. On the right, $m = n/10$, $d=5$ and $\rho = 2 \log(m) \sqrt{{d}/{n}}$. 
For every $n$ in a grid ranging from $1$ to $3\cdot10^4$, $100$ datasets were simulated. 
The solid lines give the TPR, i.e.\ the fraction of the  $100$ runs
in which the tests correctly detected the signal. The red lines correspond to $T_{m \le d}$, 
the blue lines to $T_{m \ge d}$.}
\end{figure}

In the second experiment we also consider two scenarios. 

The first scenario corresponds to a situation in which a fixed number 
of machines $m$  receive more observations as $n$ increases, but the dimension
of the signal increases as well. Specifically, we take 
 $d = n^{\frac{2}{3}}$ and $m = 500$. 
 We set the signal strength slightly above the detection limit for this case, namely
 $\rho = \log(d) (d m)^{1/4} /\sqrt{n}$. Note that  $\rho \to 0$ as $n \to \infty$. 
 In view of the theory we expect that  the test $T_{m \le d}$ 
 detects the signal consistently, whereas the public coin test $T_{m \ge d}$ should 
 have a worse performance,  since it requires that $\rho \gtrsim \sqrt{d/n}$. 

In the second scenario we take $m  = n/10$ and $d=5$ and set
$\rho = 2 \log(m) \sqrt{d/n}$,  which is again slightly above the detecting limit. 
This corresponds to the situation in which for each $10$ additional observations
a new machine is added,  whilst the dimension of the signal remains fixed. 
In this case we also have $\rho \to 0$ as $n \to \infty$.  We expect 
that for large $n$ the test $T_{m \le d}$ will fail to consistently detect the signal,  
as it requires $\rho \gtrsim {(md)^{1/4}}/{\sqrt{n}}$ to do so.

The resulting plots can be found in Figure \ref{fig: n m d seq}. 
The results are again in accordance with our theoretical findings.
We note in particular that  the public coin test $T_{m \ge d}$ 
is a suitable choice in the realistic scenario where the number of machines $m$ 
scales with $n$, so batches of additional observations are distributed over additional machines,
while the dimension of the signal  remains fixed or at least relatively small compared to the number of observations.

\section{Proof of Theorem \ref{thm_upper_bound}}\label{sec: proof:thm:UB}

Let $\alpha \in (0,1)$ be given, recall that $\bar{\kappa}_\alpha^2 = 3\log \left(4/\alpha \right)$, $\tilde{\kappa}_\alpha^2 = (1/3) \log \left(16/\alpha \right)$, $\kappa_{\alpha}^2=4/\alpha$, and set
\begin{align}
M_{\alpha}&:=\max\left\{ (2^5 5 \bar{\kappa}_{\alpha})^2,\bar{D},36\tilde{\kappa}_\alpha^2,4e^2\pi \tilde{\kappa}_{\alpha}^2\alpha^{1/2}\right\},\nonumber\\
C_\alpha&:=\max\left\{  2^4 M_{\alpha}^2 \kappa_{\alpha}^2, 2(1+\sqrt{2})\kappa_{\alpha}M_{\alpha},80\bar{\kappa}_{\alpha},2^{12}e^2\tilde{\kappa}_{\alpha}^2\alpha^{-5/2} \right\},\label{def:M_alpha:C_alpha}
\end{align}
where $\bar{D}$ is defined in Lemma \ref{claim_larger_than_half_lemma}. Then we show below that for all $d,m,n \in \N$ at least one of the three distributed tests $T_{m\leq d},T_{m \geq d},T_{m \asymp 1},$ given in \eqref{def_global_test_m_is_1}-\eqref{def_global_test_large_m}, achieves uniform consistency at level $\alpha$, i.e. 
\begin{equation}\label{consistency_display}
\mathcal{R} (H_\rho, T) \leq \alpha,
\end{equation}
for some $T \in\{ T_{m \asymp 1},T_{m\leq d},T_{m \geq d}\}$.

We distinguish three regimes in view of the interplay of $m$ and $d$, and show in the following subsections that the corresponding tests reach uniform consistency:
\begin{itemize}
\item \textbf{Case 1:} $ M_\alpha<m \leq d$.
\item \textbf{Case 2:} $ M_\alpha\vee d <m$.
\item \textbf{Case 3:} $m \leq M_\alpha$.
\end{itemize}
with corresponding tests $T_{m\leq d},T_{m \geq d},T_{m \asymp 1}$, respectively. 

\subsection{Case 1: $ M_\alpha<m \leq d$}\label{sec: case1}
In view of Chernoff's bound, see Lemma \ref{lem:chernoff}, for $m>  4\bar{\kappa}_\alpha^2 $,
\begin{align}
P_0 \left( T_{m \leq d} = 1 \right)&= P_0 \left(\sum_{j=1}^m (T_{m \leq d}^{(j)}-1/2) \geq \bar\kappa_{\alpha}/\sqrt{m} \right)\nonumber\\
& \leq 2 \exp\left(-2 \bar{\kappa}_\alpha^2/3 \right)< \alpha/2.\label{type_I_error_bound_small_m_test}
\end{align}

We now turn to bounding the type II error probability. For arbitrary $\mu =(\mu_1,\dots,\mu_d) \in H_\rho$,
\begin{equation}\label{Tj_under_Pf}
S^{(j)}_{m \leq d}=\underset{i=1}{\overset{d}{\sum}} \left( \sqrt{\frac{n}{m}}\mu_{i}+ Z_{i}^{(j)} \right)^{2}
\end{equation}
with $Z_{i}^{(j)}\stackrel{iid}{\sim}N(0,1)$, follows a noncentral chi-square distribution with $d$ degrees of freedom and noncentrality parameter $\delta:= (n/m)\| \mu\|_2^2$, i.e. $S^{(j)}_{m \leq d}\stackrel{ind}{\sim} \chi_{d}^2(\delta)$.

Let us take independent random variables $V_d^{\delta}\sim \chi_{d}^2(\delta)$ and $U_d\sim \chi_{d}^2(0)$ . Then in view of Lemma \ref{claim_larger_than_half_lemma}, for $c>1/40$ and for all $d\geq \bar{D}$,
\begin{align*}
E_\mu (F_{\chi^2_d} (S^{(j)}_{m \leq d})) = \text{Pr}\left(  V^{\delta}_d \geq U_d \right)\geq \frac{1}{2} + c \left(\frac{\delta}{\sqrt{d}} \wedge \frac{1}{2} \right),
\end{align*}
where $E_\mu$ is the expectation corresponding to $P_{\mu}$. This further yields
\begin{align*}
P_\mu & \left(  \left| \underset{j=1}{\overset{m}{\sum}} \left[T^{(j)}_{m \leq d} - \frac{1}{2}\right] \right|   \leq \sqrt{m}\bar{\kappa}_\alpha \right) \\ 
& \leq P_\mu\left( \underset{j=1}{\overset{m}{\sum}} \left[T^{(j)}_{m \leq d} - E_\mu (F_{\chi^2_d} (S^{(j)}_{m \leq d})) \right]  \leq \sqrt{m}\bar{\kappa}_\alpha \left( 1 - \frac{c\sqrt{m}}{\bar{\kappa}_\alpha}   \left(\frac{\delta}{\sqrt{d}} \wedge  \frac{1}{2} \right) \right) \right).
\end{align*}
In view of Chernoff's bound and the inequality $E_\mu (F_{\chi^2_d} (T^{(j)}_{m \leq d}))\leq 1$, this is further bounded by $2\exp(-\bar{\kappa}_\alpha^2 /3)=\alpha/2$, given that
\begin{equation}\label{chernoff_enabler}
\frac{c\sqrt{m}}{\bar{\kappa}_\alpha}   \left(\frac{\delta}{\sqrt{d}} \wedge  \frac{1}{2} \right) \geq 2.
 \end{equation}
This last inequality follows from the assumption \eqref{rho_seq_bound_that_is_uniformly_testable} 
\begin{align*}
 \sqrt{\frac{m}{d}}\delta = \frac{n\| \mu \|_2^2}{\sqrt{dm}} \geq C_{\alpha}\geq  2 \bar{\kappa}_\alpha/c\quad \text{and}\quad m > 16 \bar{\kappa}_\alpha^2/c^2.
\end{align*}

\subsection{Case 2: $ M_\alpha\vee d <m$}\label{sec: case2}

For the choice  and $m >  4\tilde\kappa_\alpha^2 $ the same bound as in \eqref{type_I_error_bound_small_m_test} holds for $P_0 \left( T_{m \geq d} = 1 \right)$. 

For the type II error, consider $\mu = (\mu_1,\dots,\mu_d) \in H_\rho$. Define for $\eta >0$, $\nu >0$ the events $D_\eta := \{ \|U\|_2/\sqrt{d} \leq \eta \}$ and
\begin{equation*}
A_{\nu,\eta} := \left \{ \frac{\sqrt{n}}{\sqrt{d}}\underset{i=1}{\overset{d}{\sum}} \mu_i U_i > \nu \right \}\cap D_\eta, \; B_{\nu,\eta} := \left \{ \frac{\sqrt{n}}{\sqrt{d}}\underset{i=1}{\overset{d}{\sum}} \mu_i U_i < -\nu\right \} \cap D_\eta,
\end{equation*} 
where $U=(U_1, \dots,U_d)\sim N(0,I_d)$ is the public coin random vector, whose probability distribution we shall denote by $Q$. We set
\begin{align}
\eta =4 \alpha^{-1/2} \quad \text{and} \quad  \nu = 4e\sqrt{2\pi}\tilde\kappa_\alpha \sqrt{\eta} ,  \label{def:eta_nu}
\end{align}
and note that $C_{\alpha} \geq 2^7 \nu^2/(\pi \alpha^2)$.

Then the type II error is bounded from above as
\begin{align}
E_Q P_\mu \left( T_{m \geq d} = 0|U=u \right) &\leq  \int_{A_{\nu,\eta}} P_\mu \left(  \underset{j=1}{\overset{m}{\sum}} (T^{(j)}_{m \geq d} - 1/2)  \leq  \sqrt{m}\tilde\kappa_\alpha \bigg| U = u  \right) dQ(u)\nonumber \\ 
&+ \int_{ B_{\nu,\eta}} P_\mu \left(  \underset{j=1}{\overset{m}{\sum}} (T^{(j)}_{m \geq d} - 1/2)  \geq  - \sqrt{m}\tilde\kappa_\alpha \bigg| U = u  \right) dQ(u)\nonumber \\
 &\qquad+ Q  (A^c_{\nu,\eta} \cap  B_{\nu,\eta}^c).\label{eq: type_II_Case3}
\end{align}
We show below that each of the first two terms on the right hand side are bounded by $\alpha/8$ and the third term by $\alpha/4$, resulting our statement.

First we deal with the third term in \eqref{eq: type_II_Case3}. Since
\begin{equation*}
\frac{\sqrt{n}}{\sqrt{d}}\underset{i=1}{\overset{d}{\sum}} \mu_i U_i \sim \mathcal{N}(0, \frac{n}{d} \|\mu\|_2^2),
\end{equation*}
the set $(A_{\nu,\eta} \cup B_{\nu,\eta})^c = A_{\nu,\eta}^c \cap B_{\nu,\eta}^c$, in view of Chebyshev's inequality, assumption \eqref{rho_seq_bound_that_is_uniformly_testable}, and definitions \eqref{def:eta_nu}, satisfies that
\begin{align*}
Q( A_{\nu,\eta}^c \cap B_{\nu,\eta}^c) &\leq \text{Pr} \left( \sqrt{\frac{n}{d}} \|\mu\|_2 |Z| \leq \nu \right) + Q(D_\eta^c) 
\leq \frac{2\nu}{\sqrt{2\pi C_{\alpha}}} + \frac{2}{\eta^2}\leq \alpha/4.
\end{align*}
where $Z\sim N(0,1)$.

Next we deal with the first term on the right hand side of \eqref{eq: type_II_Case3}. For $u\in  A_{\nu,\eta}$, we have $\|u\|_2/\sqrt{d} \leq \eta$, hence for $m\geq \nu^2/(2\eta^2)$,
\begin{align*}
\text{Pr} \left( - \frac{\nu}{\sqrt{m}} < \frac{1}{\sqrt{d}} \underset{i=1}{\overset{d}{\sum}}  u_i Z_i  \leq 0 \right) &\geq \text{Pr} \left( - \frac{\nu}{\sqrt{ m}\eta} < Z  \leq 0 \right) \\
&\geq\frac{e^{-\frac{\nu^2}{2m\eta^2}}}{\sqrt{2\pi}}  \frac{\nu}{\sqrt{ m}\eta}\geq c' \frac{\nu}{\sqrt{ m}\eta},
\end{align*}
where $Z,Z_1,\dots,Z_d\stackrel{iid}{\sim}N(0,1)$ and $c'=(e\sqrt{2\pi})^{-1} $. Therefore,
\begin{align*}
 P_\mu \left (  S_{m \geq d}^{(j)} \geq 0 \bigg| U = u \right)
&=   \text{Pr} \left( \frac{\sqrt{n}}{\sqrt{dm}}\underset{i=1}{\overset{d}{\sum}} \mu_i u_i + \frac{1}{\sqrt{d}} \underset{i=1}{\overset{d}{\sum}}  u_i Z_i  \geq 0 \right) \\
&\geq  \text{Pr} \left( \frac{\nu}{\sqrt{m}} + \frac{1}{\sqrt{d}} \underset{i=1}{\overset{d}{\sum}}  u_i Z_i  \geq 0 \right) \\
&= \frac{1}{2} + \text{Pr} \left( - \frac{\nu}{\sqrt{m}} < \frac{1}{\sqrt{d}} \underset{i=1}{\overset{d}{\sum}}  u_i Z_i  \leq 0 \right) \\
&\geq  \frac{1}{2} + c' \frac{\nu}{\sqrt{\eta m}}.
\end{align*}
Then in view of Chernoff's bound 
\begin{align*}
\int_{A_\nu}& P_\mu \left(  \underset{j=1}{\overset{m}{\sum}} (T^{(j)}_{m \geq d}  - 1/2)  \leq  \sqrt{m}\tilde\kappa_\alpha \bigg| U= u \right) d Q(u)  \\ 
&\leq 
 \int_{A_\nu} P_\mu \left(  \underset{j=1}{\overset{m}{\sum}} \left[T^{(j)}_{m \geq d} - E_\mu( T^{(j)}_{m \geq d} |U=u)\right]  \leq  \sqrt{m}\big(\tilde\kappa_\alpha -\frac{c' \nu}{\sqrt{\eta}}\big) \bigg| U= u \right) d Q(u).
\end{align*}
Since by definition \eqref{def:eta_nu} $\tilde\kappa_\alpha- c'\nu/\sqrt{\eta}=-3\tilde\kappa_\alpha$, in view of Chernoff's bound (Lemma \ref{lem:chernoff}), for $36 \tilde\kappa_\alpha^2 < m$  the preceding display is further bounded by $2e^{ - 3 \tilde\kappa_\alpha^2}=\alpha /8$.

We can deal with the second term on the right hand side of \eqref{eq: type_II_Case3} similarly. First, we obtain for $u\in  B_{\nu,\eta}$ that
\begin{equation*}
P_\mu \left (  S_{m \geq d}^{(j)} \geq 0 \bigg| U = u \right)  \leq 1/2 - c' \frac{\nu}{\sqrt{\eta m}}.
\end{equation*}
And then we can derive as above that
\begin{equation*}
E_Q \mathbbm{1}_{B_\nu} P_\mu \left(  \underset{j=1}{\overset{m}{\sum}} \left[T^{(j)}_{m \geq d} - \frac{1}{2}\right]  \geq  - \sqrt{m}\tilde\kappa_\alpha  \right) \leq \alpha /8,
\end{equation*}
concluding the proof of case 2.

\subsection{Case 3: $m \leq M_\alpha$}\label{sec: case3}
Finally, for completeness, we deal with the case when the number of machines is limited and we are back bascially to the single server, centralized case from a minimax rate point of view.

Under the null hypothesis, $(n/m) \|X^{(1)}\|_2^2$ follows a chi-square distribution with $d$ degrees of freedom, so by Chebyshev's inequality
\begin{equation*}
 P_0 \left( T_{m \asymp 1} = 1 \right) \leq 2/\kappa_\alpha^2 = \alpha/2.
\end{equation*}

 For any $\mu \in H_{\rho}$, it holds that
\begin{align*}
P_\mu \big( \frac{n}{\sqrt{d}m}& \|X^{(1)}\|_2^2 - \sqrt{d} \leq\kappa_\alpha \big)\\
& = \text{Pr}\left( \left| \frac{n}{\sqrt{d} m} \|\mu\|_2^2 + 2d^{-1/2} \underset{i=1}{\overset{d}{\sum}} \mu_i Z_i + d^{-1/2} \underset{i=1}{\overset{d}{\sum}} (Z_i^2 - 1) \right| \leq \kappa_\alpha \right)
\end{align*}
where $Z_i\stackrel{iid}{\sim} N(0,1)$, $i=1,...,d$. By the reverse triangle inequality, we can bound the previous probability with
\begin{align*}
\text{Pr}\left( \left|  2d^{-1/2} \underset{i=1}{\overset{d}{\sum}} \mu_i Z_i \right| \geq \frac{n\|\mu\|_2^2}{2\sqrt{d} m}  \right) + \text{Pr}\left(  \left|  d^{-1/2} \underset{i=1}{\overset{d}{\sum}} (Z_i^2 - 1)  \right| \geq   \frac{n \|\mu\|_2^2}{2 \sqrt{d}  m} -\kappa_\alpha  \right).
\end{align*}

In view of Chebyshev's inequality and $\|\mu\|_2^2 \geq  \rho^2 \geq C_{\alpha} \sqrt{d}/n$,
\begin{equation*}
\text{Pr}\left( \left|  2d^{-1/2} \underset{i=1}{\overset{d}{\sum}} \mu_i Z_i \right| \geq \frac{n\|\mu\|_2^2}{2\sqrt{d} m}  \right)\leq \frac{16 m^2}{n^2 \|\mu\|_2^2} \leq  \frac{16 m^2}{C_{\alpha} } \leq \frac{\alpha}{4}
\end{equation*}
for $C_\alpha \geq 2^6 M_\alpha^2/\alpha$. Likewise, for $C_\alpha\geq {2(1+\sqrt{2}) \kappa_\alpha M_\alpha}$
\begin{equation*}
  \frac{n \|\mu\|_2^2}{2 \sqrt{d}  m} - \kappa_\alpha \geq  \sqrt{2} \kappa_\alpha
\end{equation*}
and consequently, by applying again Chebyshev's inequality, the second term is bounded by
\begin{equation*}
\text{Pr}\left(  \left|  d^{-1/2} \underset{i=1}{\overset{d}{\sum}} (Z_i^2 - 1)  \right| \geq \sqrt{2} \kappa_\alpha   \right) \leq \frac{1 }{\kappa_\alpha^2} \leq \frac{\alpha}{4}.
\end{equation*}
Hence we can conclude that for $C_{\alpha}\geq \max\{2^6 M_\alpha^2/\alpha,   {2(1+\sqrt{2}) \kappa_\alpha M_\alpha}\}$,
\begin{equation*}
P_\mu \big(T_{m \asymp 1}=0 \big)\leq \alpha/2,
\end{equation*}
concluding the proof of the theorem.

\section{Lemmas for Theorem \ref{thm_upper_bound}}\label{sec: thm_upper_bound}
In this section we collect the proofs of the lemmas used to derive the minimax testing upper bound.

\begin{lemma}\label{claim_larger_than_half_lemma}

Let $U_d$ and $V_d^{\delta_d}$ be independent chi-square distributed random variables with $d$ degree of freedom and non-centrality parameters zero and $\delta_d>0$, respectively. Then for a universal $\bar{D}\in\mathbb{N}$, not depending on $\delta_d$, we have for all $d\geq \bar{D}$ that
\begin{equation}\label{lemma_noncent_cent_chisq}
\text{Pr}\left(  V_d^{\delta_d} - U_d  \geq 0 \right) \geq \frac{1}{2} + \frac{1}{40}( \frac{{\delta_d}}{\sqrt{d}} \wedge \frac{1}{2}).
\end{equation}

\end{lemma}
\begin{proof}

First note that the function $\delta \mapsto \text{Pr}\left(  V_d^{\delta} - U_d  \geq 0 \right)$ is  monotone increasing. Then
$$\text{Pr}\left(  V_d^{\delta_d} - U_d  \geq 0 \right) \geq \text{Pr}\left(  V_d^{\delta_d\wedge \sqrt{d}/2} - U_d  \geq 0 \right),$$
so without loss of generality we can assume that $\delta_d \leq \sqrt{d}/2$.

The density of $V_d^{\delta_d}$ is 
\begin{equation*}
\underset{k=0}{\overset{\infty}{\sum}} \frac{e^{- {\delta_d}/2} \left( {\delta_d}/2\right)^k}{k!} p_{d+2k},
\end{equation*}
where $p_k$ denotes the $\chi^2_k$-density.
By the independence of $U_d$ and $V_d^{\delta_d}$,
\begin{align*}
\text{Pr}\left(  V_d^{\delta_d} - U_d  \leq 0 \right) &= \underset{k=0}{\overset{\infty}{\sum}} \frac{e^{- {\delta_d}/2} \left( {\delta_d}/2\right)^k}{k!} \int_{\{ v-u \geq 0\}} p_{d+2k}(v) p_d (u) d(v,u).
\end{align*}
Let $U_d' \sim \chi^2_d$ and $U_{2k}'' \sim \chi^2_{2k}$ be independent from each other and $U_d$. For any given $k \in \N$, we have
\begin{align*}
\int_{\{ v-u \geq 0\}} p_{d+2k}(v) p_d (u) d(v,u) =\text{Pr} \left( U_d - U_d' \leq U_{2k}'' \right).
\end{align*}

For convenience let us introduce the notation $W_d=(U_d-U_d')/(2\sqrt{d})$. Conditioning and using independence once more, the latter equals
\begin{equation*}
\int \text{Pr}\left( W_d \leq  \frac{u}{2\sqrt{d}} \right) d \P_{U_{2k}''}(u) = \frac{1}{2} + \int \text{Pr}\left( 0 \leq W_d  \leq  \frac{u}{2\sqrt{d}} \right) d \P_{U_{2k}''}(u).
\end{equation*}
Since $U_{2k}''$ has a median larger than $2k/3$ and the map $u \mapsto \text{Pr}\left( 0 \leq W_d  \leq  \frac{u}{2\sqrt{d}} \right)$ is increasing, we have that the second term in the last display satisfies
\begin{align*}
\int \text{Pr}\left( 0 \leq W_d  \leq  \frac{u}{2\sqrt{d}} \right)  p_{{2k}}(u)du &\geq \text{Pr}\left( 0 \leq W_d\leq  \frac{k}{3\sqrt{d}} \right) \int_{[\frac{2k}{3},\infty)} p_{{2k}}(u)du \\  &\geq \frac{1}{2}  \text{Pr}\left( 0 \leq W_d  \leq  \frac{k}{3\sqrt{d}} \right).
\end{align*}
By combining the above inequalities we obtain that
\begin{equation}\label{half_breakdown}
\text{Pr}\left(  V_d^{\delta_d} - U_d  \leq 0 \right) \geq \frac{1}{2} + \frac{1}{2} \underset{k=0}{\overset{\infty}{\sum}} \frac{e^{- {\delta_d}/2} \left( {\delta_d}/2\right)^k}{k!}    \text{Pr}\left( 0 \leq W_d  \leq  \frac{k}{3\sqrt{d}} \right).
\end{equation}

Assume now that $\delta_d \gtrsim 1$. Let $k_d$ be the largest integer such that $ k_d \leq 3\sqrt{d}$. We divide the sum on the right hand of the preceding display to two parts, i.e. $k< k_d$ and $k\geq k_d$. By applying Lemma \ref{chi_sq_nens_unity_lemma} with $\varepsilon_d= k$, it holds that for $c_0=e^{-9/8}/6$,
\begin{align*}
 \underset{k=0}{\overset{k_d}{\sum}} \frac{e^{- {\delta_d}/2} \left( {\delta_d}/2\right)^k}{k!}    \text{Pr}\left( 0 \leq W_d  \leq  \frac{k}{3\sqrt{d}} \right) &\geq \frac{c_0 }{ \sqrt{d}} \underset{k=1}{\overset{k_d}{\sum}} \frac{e^{- {\delta_d}/2} \left( {\delta_d}/2\right)^{k}}{(k-1)!} \\
&\geq \frac{c_0 {\delta_d}}{2 \sqrt{d}} \underset{k=0}{\overset{k_d-1}{\sum}} \frac{e^{- {\delta_d}/2} \left( {\delta_d}/2\right)^{k}}{k!}.
\end{align*}
We have $\text{Pr}\left( 0 \leq W_d \leq 1 \right) \stackrel{d}{\to} \text{Pr}\left( 0 \leq Z \leq  1 \right) > 1/3$, hence there exists a $D_1\in\mathbb{N}$, such that for all $d\geq D_1$ we have $\text{Pr}\left( 0 \leq W_d \leq 1 \right)> 1/3$. For $k > k_d$ we have $k > 3\sqrt{d}$, hence for all $d\geq D_1$,
\begin{align*}
 \underset{k> k_d}{\overset{\infty}{\sum}} \frac{e^{- {\delta_d}/2} \left( {\delta_d}/2\right)^k}{k!}    \text{Pr}\left( 0 \leq W_d  \leq  \frac{k}{3\sqrt{d}} \right) &\geq \frac{c_0}{2} \underset{k> k_d}{\overset{\infty}{\sum}} \frac{e^{- {\delta_d}/2} \left( {\delta_d}/2\right)^k}{k!}. 
\end{align*}
Since ${\delta_d}/\sqrt{d}\leq 1/2$, we have for $d\geq D_1$,
\begin{align*}
\frac{1}{2} \underset{k=0}{\overset{\infty}{\sum}} \frac{e^{- {\delta_d}/2} \left( {\delta_d}/2\right)^k}{k!}    \text{Pr}\left( 0 \leq W_d  \leq  \frac{k}{3\sqrt{d}} \right) \geq 
\frac{c_0 {\delta_d}}{2 \sqrt{d}} (1- \frac{e^{-{\delta_d}/2} \left({\delta_d}/2\right)^{k_d}}{k_d!} ).
\end{align*}

The proof is finished by showing that for large enough $d$ we have $c_0/2-1/40>  ({\delta_d}/2)^{k_d}/k_d! >0$. Recalling that $2\sqrt{d}\leq 3\sqrt d-1\leq k_d \leq3\sqrt d$ and hence $\delta_d\leq \sqrt{d}/2\leq \sqrt{d}/4$ we get in view of Stirling's inequality, that for some universal constant $C>0$
\begin{align*}
\frac{\left({\delta_d}/2\right)^{k_d}}{k_d!} \leq \frac{\left(k_d / 4\right)^{k_d}}{k_d!} \lesssim e^{k_d(1 - \log 4)} k_d^{-1/2},
\end{align*}
which is bounded from above by $c_0/2-1/40$ for $d\geq D_1$, for some sufficiently large $D_1>0$.

\end{proof}

\begin{lemma}\label{chi_sq_nens_unity_lemma}
Let $U_d,U_d'\stackrel{iid}{\sim} \chi^2_d$, and $0<\varepsilon_d\leq C\sqrt{d}$. Then there exists a large enough $D_0\in\mathbb{N}$, such that for all $d\geq D_0$
\begin{equation*}
\text{Pr}\left( 0 \leq \frac{U_d - U_d'}{2\sqrt{d}} \leq  \frac{\varepsilon_d }{\sqrt{d}} \right) \geq  \frac{e^{-C^2/8}}{6}  \frac{\varepsilon_d }{\sqrt{d}}. 
\end{equation*}
\end{lemma}
\begin{proof}
The characteristic function of the random variable $W_d := (U_d - U_d')/(2\sqrt{d})$ is
\begin{align*}
\phi_d(t) &= \E e^{i tW_d}= \E e^{i \frac{t}{2\sqrt{d}} U_d} \E e^{-i \frac{t}{2\sqrt{d}} U_d'} \\
 &= ( 1 + i t/\sqrt{d} )^{-d/2} ( 1 - i t/\sqrt{d} )^{-d/2}\\
&= ( 1 + t^2/d)^{-d/2}\stackrel{d\rightarrow\infty}{\longrightarrow}e^{-t^2/2}. 
\end{align*}

Using the Fourier inversion formula, the  density $f_{W_d}$ of $W_d$ satisfies
\begin{align*}
f_{W_d } (v) &= \frac{1}{2 \pi} \int_\R e^{itv} \phi_d(t) dt = \frac{1}{2 \pi} \int_\R \cos (tv) \phi_d(t) dt,
\end{align*}
where the second equality follows from the symmetry of $\phi_d$. Let 
\begin{align*}
g(v)&:=\frac{1}{2 \pi} \int_\R \cos (tv) e^{-t^2/2}dt= \frac{1}{\sqrt{2\pi}} e^{-v^2/2},
\end{align*}
where the last equation follows for instance by contour integration. Then by the dominated convergence theorem
\begin{align*}
| f_{W_d } (v) - g(v)| &\leq \frac{1}{2 \pi} \int_\R | e^{-t^2/2} - \phi_d(t)| dt\stackrel{d\to\infty}{\longrightarrow}0.
\end{align*}

By the earlier established uniform convergence we have that for every $d\geq D_0$, for some large enough $D_0$, 
\begin{equation*}
\int_0^{\frac{\varepsilon_d }{2\sqrt{d}}} f_{W_d} (v) dv \geq \frac{1}{3} e^{-\varepsilon_d^2/(8d)} \frac{\varepsilon_d }{2\sqrt{d}}=\frac{e^{-C^2/8}}{6} \frac{\varepsilon_d }{\sqrt{d}},
\end{equation*}
where the constant $1/3$ is arbitrary and could be taken anything smaller than $1/\sqrt{2\pi}$.
\end{proof}

\begin{lemma}[Chernoff's bound]\label{lem:chernoff}
Let $B_i\stackrel{ind}{\sim} \text{Ber}(p_i)$, $i=1,...,k$, and $0<\delta<1$. Then
\begin{equation}\label{Chernoff_bound}
\text{Pr} \left( \left| \underset{i=1}{\overset{k}{\sum}} (B_i - \E B_i) \right| \geq \delta \underset{i=1}{\overset{k}{\sum}} p_i  \right) \leq 2 e^{-(\delta^2/3) \sum_{i=1}^{k}  p_i }.
\end{equation}
\end{lemma}

\section{Lemmas for Theorem \ref{lower_bound_thm}}
In this section we collect the proofs of the lemmas used to derive the minimax testing lower bound.

\subsection{Proof of Lemma \ref{lemma_mutual_information_testing_lb}}\label{proof: lem1}

In view of \eqref{starting_pt_display_proof_mut_info_lem} we have
\begin{align*}
\mathcal{R}_{}(H_{\rho},T) &\geq 1 - \left( \P( T = 0 | V = 0) - \P( T = 0 | V = 1) \right)  \\
&\geq 1 - \left|\P^{T|V=0} - \P^{T|V=1}\right|( T = 0 ) \\
&\geq 1 - \| \P^{T|V=0} - \P^{T|V=1} \|_{TV}.
\end{align*}
By the triangle inequality,
\begin{equation*}
\| \P^{T|V=0} - \P^{T|V=1} \|_{TV} \leq \| \P^{T|V=0} - \P^{T} \|_{TV} + \| \P^{T} - \P^{T|V=1} \|_{TV}.
\end{equation*}
Applying the second Pinsker bound to the two terms on the RHS and using that $2ab \leq a^2 + b^2$,
\begin{align*}
\| \P^{T|V=0} - \P^{T|V=1} \|_{TV}^2 &\leq  D_{KL}(\P^{T|V=0} \| \P^{T}) + D_{KL}(\P^{T|V=1} \| \P^{T}) \\
&= 2 I_\Pi(V,T),
\end{align*}
which completes the proof of the lemma.

\subsection{Proof of Lemma \ref{lem: tensor} }\label{proof: lem_tensor}
We prove a more general lemma, but before stating it we recall some information theoretic definitions and identities, see \cite{cover:2012,polyanskiy:2014}. 

For discrete random variables $X$ and arbitrary random variable $Y$,  define the \emph{entropy of} $X$ as 
\begin{equation*}
H(X) = - \underset{x}{\overset{}{\sum}} \P(X=x) \log \P(X=x)
\end{equation*}
and the \emph{conditional entropy of $X$ given $Y$}, 
\begin{equation*}
H(X|Y) = E_Y H(X|Y=y)=-E_Y \underset{x}{\overset{}{\sum}} \P(X=x|Y=y) \log \P(X=x|Y=y).
\end{equation*}
We also recall that conditioning reduced entropy $H(X)\geq H(X|Y)$. Following from this conditioning, on an arbitrary random vector $Z$, also reduces conditional entropy
 \begin{equation*}
H(X|Y) =\int H(X|Y=y)dP_Y(y)\geq \int H(X|Y=y,Z)dP_Y(y)=H(X|Y,Z).
\end{equation*}

For random variables $X,Y, Z$ we define the mutual information between $X$ and $Y$ and conditional mutual information between $X$ and $Y$ given $Z$ as
\begin{align*}
I(X;Y)&=D_{KL}(P_{(X,Y)}\| P_{X}\times P_{Y}),\\
I(X;Y|Z=z)&=D_{KL}(P_{(X,Y)|Z=z}\| P_{X|Z=z}\times P_{Y|Z=z}),\\
I(X;Y|Z)&=\int I(X;Y|Z=z)dP_Z(z).
\end{align*}

Next we recall some conditions of the mutual information. First we note that $I(X,Y)=0$ if and only if $X$ is independent from $Y$.
The chain rule for the mutual information between the random vector $Y=(Y^{(1)},...,Y^{(m)})$ and $V$ is
\begin{align}
I(V;Y)=\sum_{j=1}^{m} I(V;Y^{(j)}|Y^{(1)},..., Y^{(j-1)}).\label{eq: chain_rule}
\end{align}
For discrete random variable $X$ and arbitrary random variable $Y$
\begin{align}
I(X;Y)&=E_{(X,Y)}\log \frac{dP_{(XY)}}{dP_X dP_Y}\\
&= E_{(X,Y)}\log  \frac{1}{dP_X}-E_{(X,Y)}\log  \frac{1}{dP_{(X|Y=y)}}\\
&=H(X)-H(X|Y).\label{eq: mutinf}
\end{align}
In addition, by similar arguments, for arbitrary random variable $Z$ we have
\begin{align}
I(X;Y|Z)=H(X|Z)-H(X|Y,Z).\label{eq: cond_mutinf}
\end{align}

\begin{lemma}\label{lemma_dep_mutual_info_tensor}
Let us assume that the discrete random variable $V$ { and the discrete random vector $F$ are such that the pair $(V,F)$ is independent from the random variable $U$ and the discrete random vector $Y=(Y_1,...,Y_m)$ satisfies that $Y_j$ is conditionally independent from $Y_{1:j-1}:=(Y_1,...,Y_{j-1})$ given $U$ and $(V,F)$}, then
\begin{equation*}
I(V;Y) \leq \underset{j=1}{\overset{m}{\sum}}I(V;Y_j|U) {+ \underset{j=1}{\overset{m}{\sum}}I(F;Y_j|U, V)}. 
\end{equation*}
\end{lemma}

\begin{proof}
A non-public coin version of the lemma is given for instance in \cite{raginsky:2016}.

First note that in view of \eqref{eq: mutinf} and since conditioning reduces entropy
\begin{align*}
I\big((Y,U);V\big) = H(V)-H(V|Y,U)\geq H(V)-H(V|Y)=I(Y;V).
\end{align*}
Furthermore, by the chain rule \eqref{eq: chain_rule} and the independence of $U$ and $V$,
\begin{align*}
I\big((Y,U);V\big)=I\big(Y;V|U\big) + I(U;V)=I(Y;V|U).
\end{align*}
{Similarly, by the chain rule and nonnegativity of mutual information,
\begin{equation*}
I\big(V;Y|U\big)  = I\big((V,F);Y|U\big) - I\big(F;Y|U,V\big) \leq  I\big((V,F);Y|U\big).
\end{equation*}
By the identity \eqref{eq: cond_mutinf} and the chain rule \eqref{eq: chain_rule},
\begin{align*}
I\big((V,F);Y|U\big) &= H(Y|U) - H(Y|V,F,U) \\
&= \underset{j=1}{\overset{m}{\sum}} H(Y_j | Y_{1:j-1},U) - H(Y_j |V,F, Y_{1:j-1},U).
\end{align*}
Since conditioning reduces entropy we have $H(Y_j | Y_{1:j-1},U) \leq H(Y_j|U)$. Furthermore, by the conditional independence of  $Y_{1:j-1}$ and $Y_j$ given $(U,V,F)$ results in $H(Y_j |V,F, Y_{1:j-1}, U) = H(Y_j |V,F, U)$. Using these two facts, we obtain that
\begin{align*}
I\big((V,F);Y|U\big) &\leq \underset{j=1}{\overset{m}{\sum}} H(Y_j |,U)  - H(Y_j |V,F,U) \\
&= \underset{j=1}{\overset{m}{\sum}} I\big((V,F);Y_j|U\big).
\end{align*}
Combining the above displays and again applying the chain rule we now obtain that
\begin{equation*}
I(Y;V) \leq \underset{j=1}{\overset{m}{\sum}} I\big((V,F);Y_j|U\big) = \underset{j=1}{\overset{m}{\sum}} \left[I\big(V;Y_j|U\big) + I\big(F;Y_j|U,V\big)\right].
\end{equation*}}
\end{proof}

\subsection{Proof of Lemma \ref{SDPI} }\label{proof: PC-SDPI}
Using the sub-Gaussianity of the likelihood $P_{\Pi}=\int P_{\mu}d\Pi(\mu)$ verified in Lemma \ref{subgauss_lemma}, we adapt the proof of Theorem 3.7 in \cite{raginsky:2016} to the present continuous setting with public coin protocol.

\begin{proof}

We start by noting that if $48 \beta \geq 1$, the result follows immediately from the regular data processing inequality for mutual information. 

In view of the definition of the conditional mutual information and noting that $\P^{V|U=u}(v) = \P^V(v) = 1/2$ by independence of $U$ and $V$,
\begin{align}
I(V,{T}^{(j)}|U) =\frac{1}{2} \int   \sum_{v\in\{0,1\}}D_{KL}\left(\P^{{T}^{(j)}|(V,U)=(v,u)} ;  \P^{{T}^{(j)}|U=u} \right)  d\P^U(u). \label{write_out_def}
\end{align}
By Lemma \ref{chi_sq_div_bounds_KL} below,
\begin{equation}\label{entropy_variance_bound}
 D_{KL}\left(\P^{{T}^{(j)}|(V,U)=(v,u)} \|  \P^{{T}^{(j)}|U=u} \right) \leq \underset{{t} \in {{\{ 0,1 \}}}}{\overset{}{\sum}} \P^{{T}^{(j)}|U=u}({t}) \left( \frac{\P^{{T}^{(j)}|(V,U)=(v,u)}({t})}{\P^{{T}^{(j)}|U=u}({t})} - 1 \right)^2
\end{equation}
$\P^U$-almost surely. Furthermore, by Bayes rule,
\begin{equation}\label{radon_Xj_cond_Yj}
 \frac{\P^{{T}^{(j)}|(U,X^{(j)})=(u,x)}({t})}{\P^{{T}^{(j)}|U=u}({t})}=\frac{d\P^{X^{(j)}|(U,{T}^{(j)})=(u,{t})}}{d\P^{X^{(j)}|U=u}}(x)=:g_{{t},u}(x) ,
\end{equation}
where the equality holds in an $L_1\left(\P^{X^{(j)}}\right)$ sense $\P^U$-almost surely.

For $v\in\{0,1\}$, define
\begin{equation*}
\mathscr{L}_v (X^{(j)}) :=  \frac{d\P^{X^{(j)}|V=v} }{d\P^{X^{(j)}}}(X^{(j)}).
\end{equation*}
Since $V \to (X^{(j)},U) \to {T}^{(j)}$ forms a Markov chain, we can write
\begin{align*}
\frac{\P^{{T}^{(j)}|(V,U)=(v,u)}({t})}{\P^{Y^{(j)}|U=u}({t})}  &=  \int  \frac{\P^{{T}^{(j)}|(U,X^{(j)})=(u,x)}({t})}{\P^{{T}^{(j)}|U=u}({t})}   d\P^{X^{(j)}|V=v}(x) \\
&=  \E_{X^{(j)}} \left[ g_{{t},u}(X^{(j)}) \mathscr{L}_v (X^{(j)}) \right].
\end{align*}
Then in view of $\E_{X^{(j)}} \mathscr{L}_v (X^{(j)}) =1= \E_{X^{(j)}}  g_{y,u}(X^{(j)})  $, $\P^U$-a.s., the right hand side of \eqref{entropy_variance_bound} equals
\begin{equation*}
\underset{{t} \in {\{ 0, 1 \}}}{\overset{}{\sum}}  \P^{{T}^{(j)}|U=u}({t}) \text{Cov}\left( \mathscr{L}_v (X^{(j)}),g_{{t},u}(X^{(j)})\right)^2, \;\;\; \P^U\text{-a.s.}
\end{equation*}
By Theorem 4.13 in \cite{boucheron:etal:2013}, we have that
\begin{equation*}
\E G H \leq \E H \log H +  \log \E e^G
\end{equation*}
for any random variables $G,H$ with $\E H = 1$ and $\E e^G < \infty$. Applying this to $G = s(\mathscr{L}_v (X^{(j)}) -1 )$, $H = g_{{t},u}(X^{(j)})$ we obtain
\begin{align*}
s\text{Cov} \left(\mathscr{L}_v(X^{(j)}), g_{{t},u}(X^{(j)})\right) \leq  \log &\, \E_{X^{(j)}} \left[ e^{s (\mathscr{L}_v(X^{{(j)}}) -1)} \right] \\ &+ \E_{X^{(j)}} \left[ g_{{t},u}(X^{(j)}) \log g_{{t},u}(X^{(j)}) \right].
\end{align*}
By display \eqref{radon_Xj_cond_Yj} and the independence of $X^{(j)}$ and $U$,
\begin{equation*}
\E_{X^{(j)}} \left[ g_{{t},u}(X^{(j)}) \log g_{{t},u}(X^{(j)}) \right] = D_{KL}(\P^{X^{(j)}|(U,{T}^{(j)})=(u,{t})}\|\P^{X^{(j)}|U=u}).
\end{equation*}
Furthermore, in view of Lemma \ref{subgauss_lemma}, $\mathscr{L}_v(X^{(j)})$ is $\sqrt{24\beta}$-sub-Gaussian, hence
\begin{equation*}
 \log \E_{X^{(j)}} \left[ e^{s (\mathscr{L}_v(X^{(j)}) -1)} \right] \leq 24 \beta s^2/2.  
\end{equation*}
Taking $s = (24 \beta)^{-1} \text{Cov} \left(\mathscr{L}_v(X^{(j)}), g_{{t},u}(X^{(j)})\right)$ and combining the above displays we obtain
\begin{align*}
\text{Cov} \left(\mathscr{L}_v(X^{(j)}), g_{{t},u}(X^{(j)})\right)^2 \leq \frac{1}{2} &\text{Cov} \left(\mathscr{L}_v(X^{(j)}), g_{{t},u}(X^{(j)})\right)^2 \\ &+ 24 \beta D_{KL}(\P^{X^{(j)}|(U,{T}^{(j)})=(u,{t})}\|\P^{X^{(j)}|U=u}),
\end{align*}
which holds $\P^U$-almost surely. We now have shown that
\begin{align*}
&\int D_{KL}\left(\P^{{T}^{(j)}|(V,U)=(v,u)} \|  \P^{{T}^{(j)}|U=u} \right)  d\P^U(u)  \\ 
&\qquad\qquad \leq 48 \beta  \int \underset{{t} \in {\{ 0, 1 \}}}{\overset{}{\sum}} \P^{{T}^{(j)}|U=u}({t}) D_{KL}(\P^{X^{(j)}|(U,{T}^{(j)})=(u,{t})}\|\P^{X^{(j)}|U=u}) d\P^U(u) \\
 &\qquad\qquad=48 \beta I(X^{(j)},{T}^{(j)}|U),
\end{align*}
which in view of  \eqref{write_out_def}  concludes the proof.
\end{proof}

\subsection{Sub-Gaussianity lemma}

First we recall some notations from Section \ref{sec: lower}. Let us denote by $\Pi$ the distribution of the random vector $\epsilon R$, where $R=(R_1,\dots,R_d)$ has independent Rademacher marginals and $\epsilon>0$ is small (it is taken to be $\epsilon=\rho/\sqrt{d}$). We take $V\sim Ber(1/2)$ and set $X|(V=0)\sim N(0,\sigma^2 I_d)$ and $X|(V=1)\sim P_{\Pi}$, where $P_{\Pi}=\int P_{\mu}d \Pi(\mu)$ and $P_{\mu}$ is a multivariate Gaussian distribution with mean $\mu$ and $\sigma^2$ times the identity variance. Let $\P^{X}$ and $\P^{X|V}$ denote the corresponding  distributions of $X$ and $X|V$.

\begin{definition}
A random variable $X$ is called $\beta$\emph{-sub-Gaussian} if for all $s \in \R$, 
\begin{equation*}
\E e^{s(X-\E X)} \leq e^{\beta^2 s^2/2}.
\end{equation*}
\end{definition}

The lemma below shows that the likelihood ratios $\frac{d\P^{X|V=0} }{d\P^X}(X)$ and $\frac{d\P^{X|V=1} }{d\P^X}(X)$ are sub-Gaussian.

\begin{lemma}\label{subgauss_lemma}
 The likelihood ratios
 		\begin{equation*}
 	\frac{d\P^{X|V=0} }{d\P^X}(X) \;\; \text{ and } \;\; \frac{d\P^{X|V=1} }{d\P^X}(X)
 	\end{equation*}
 	are $\sqrt{24\beta}$-sub-Gaussian with
 	\begin{equation}\label{def:beta:subgauss}
\beta = \begin{cases}
d \epsilon^4/\sigma^4, \;\;\; \text{ if } \sigma^2/\epsilon^2 < d/2, \\
2 \epsilon^2/\sigma^2, \;\;\; \text{ if }  \sigma^2/\epsilon^2 \geq d/2.
\end{cases}
\end{equation}
\end{lemma}

\begin{proof}
Using the notation
\begin{equation*}
\mathscr{L}_v (X) :=  \frac{d\P^{X|V=v} }{d\P^{X}}(X), \quad  v\in\{0,1\},
\end{equation*}
we show below that for all $t \in \R$,
\begin{equation*}
\E_{X} e^{t(\mathscr{L}_v(X) - \E_X \mathscr{L}_v(X))} \leq e^{ 24 \beta t^2/2 }.
\end{equation*}
This is implied by
\begin{equation}\label{to_show_subgaussianity}
\P^X \left( |\mathscr{L}_v - \E_{X} \mathscr{L}_v | \geq s \right) \leq  12 \exp \left( -\frac{s^2}{2\beta} \right) \,\,\, \text{ for all } s>0,
\end{equation}
where the equivalence is well known, but a proof can be found in Lemma \ref{lem: prob_subg_implies_lapl_subg}. Since $|\mathscr{L}_v(X) - \E_{X} \mathscr{L}_v(X) | = |\mathscr{L}_v(X) - 1 | \leq 1$, it is enough to consider $0<s<1$.

To prove \eqref{to_show_subgaussianity}, let us first introduce the notation $L := \frac{dP_\pi}{dP_0}$, and note that
\begin{equation*}
\mathscr{L}_0 = \frac{2}{1 + L}\quad\text{and}\quad \mathscr{L}_1 = \frac{2}{1 + L^{-1}}.
\end{equation*}
 Then for $x \in \{ \mathscr{L}_0 - 1  \geq s \}$ we have
\begin{equation*}
\frac{2}{1+L}(x) = \mathscr{L}_0(x) \geq s+1 \quad \text{ and }\quad 0 \leq  \frac{2L}{1+L}(x) = 1 - \frac{1 - L}{1 + L}(x)  \leq 1 - s,
\end{equation*}
where the last inequality follows from $\mathscr{L}_0 - 1 = \frac{1 - L}{1 + L}$.
Consequently, $L^{-1}(x) \geq \frac{s+1}{1-s}$. Similarly, for $x \in \{ \mathscr{L}_0 - 1  \leq -s \}$,
\begin{equation*}
0 \leq \frac{2}{1+L}(x)  \leq 1-s  \text{ and } \frac{2L}{1+L}(x)  \geq 1 + s
\end{equation*}
and thus $L^{}(x) \geq \frac{s+1}{1-s}$. Combining the above bounds results in for $x \in \{ |\mathscr{L}_0 - 1| \geq s \}$ that
\begin{equation*}
|\log L(x)| \geq  \log \left(\frac{1+s}{1-s}\right) \geq \frac{2s}{1+s} \geq s,
\end{equation*}
where the last two inequalities follow from $ \log x \geq 1 - \frac{1}{x} $ and $0 < s < 1$.

Through the same computation, the above display is also true for $x \in \{ |\mathscr{L}_1 - 1| \geq s \}$. Consequently, for  $v=0,1$,
\begin{align*}
\P^X \left( |\mathscr{L}_v - \E \mathscr{L}_v | \geq s \right)& \leq  \P^X \left( |\log L| \geq s \right)\\
&=\frac{1}{2} P_0 \left( | \log (L)| \geq s \right) +  \frac{1}{2}P_\pi \left( | \log (L)| \geq   s \right).
\end{align*}

Using Markov's inequality the terms on the right hand side can be further bounded as
\begin{align}
P_0 \left( |\nu \log (L)| \geq \nu s \right) &\leq e^{-\nu s} (\E_{X|V=0} L^\nu + \E_{X|V=0} L^{- \nu}) \;\;\;\;\;\;\;\;\;\;\, \text{ for } \nu>0, \label{ineq_1} \\
P_\pi \left( |\lambda \log (L)| \geq  \lambda s \right) &\leq e^{-\lambda_1 s} \E_{X|V=1} L^{\lambda_1} + e^{-\lambda_2 s} \E_{X|V=1} L^{- \lambda_2} \; \text{ for } \lambda_1,\lambda_2 >0. \label{ineq_2}
\end{align}
Note that $ \E_{X|V=1} L^{\lambda} = \E_{X|V=0} L^{\lambda+1}$, hence choosing $\lambda_1 = \nu  - 1$ and $\lambda_2 = \nu + 1$ in display \eqref{ineq_2}, we get by combining the above two displays that for $v\in\{0,1\}$,
\begin{equation}\label{eq: UB:help}
\P^X \left( |\mathscr{L}_v- \E \mathscr{L}_v | \geq s \right) \leq \frac{1}{2} \left[  1 + e^{-s} + e^{s} \right] e^{-\nu s} (\E_{X|V=0} L^\nu + \E_{X|V=0} L^{- \nu}).
\end{equation}

We proceed by bounding the expectations in the above display after which minimizing in $\nu$ gives us the result of the lemma. Recall that $X|(V=0) \sim \mathcal{N}(0, \sigma^2I_d)$ and $X_i|(V=1) \stackrel{iid}{\sim} \frac{1}{2}\mathcal{N}(\epsilon, \sigma^2) + \frac{1}{2}\mathcal{N}(-\epsilon , \sigma^2)$, $i=1,...,d$.
Consequently,
\begin{align}
L(X) &= \underset{i=1}{\overset{d}{\Pi}} \left[ \frac{\exp\big(-\frac{1}{2\sigma^2}(X_{i} - \epsilon)^2\big) + \exp\big(-\frac{1}{2\sigma^2}(X_{i} + \epsilon)^2\big) }{2\exp(-\frac{1}{2\sigma^2}X_{i}^2)} \right] \nonumber \\
&= \underset{i=1}{\overset{d}{\Pi}} \exp(-\frac{1}{2} \epsilon^2/\sigma^2) \cosh (X_{i} \epsilon / \sigma^2). \label{L}
\end{align} 
Then by independence of $X_j$, $j=1,...,d$
\begin{align*}
\E_{X|V=0} L^\nu &= \left(e^{-\frac{\nu}{2} \epsilon^2/\sigma^2} \E  \cosh^\nu\left(\frac{\epsilon}{\sigma} Z\right) \right)^d,
\end{align*}
where $Z\sim N(0,1)$.

Next we distinguish two cases. Suppose first that $2/d \leq \epsilon^2/\sigma^2$. Let us take $\nu = s \sigma^4/(d \epsilon^4)$. Then $\nu \epsilon^2/\sigma^2 < 1/2$, as $ 0<s<1$, and hence in view of Lemma \ref{lemma_coshZ_expectation}
\begin{align*}
e^{-s\nu}(\E_{0} L^\nu + \E_{0} L^{-\nu})
&\leq \exp \left(\nu^2 \frac{d\epsilon^4}{2\sigma^4}-s\nu\right)\left( 1  + e^{(3/2)\nu d \epsilon^4/\sigma^4} \right)\\
 &\leq   \exp \left(-\frac{s^2}{2}\frac{\sigma^4}{d\epsilon^4} \right)\left( 1+e^{(3/2)s} \right).
\end{align*}


The remaining case is when $2/d >\epsilon^2/\sigma^2$. Choosing  $\nu = s\sigma^2/(2\epsilon^2)$ results in $\nu \epsilon^2/\sigma^2 < 1/2$, hence again in view of Lemma \ref{lemma_coshZ_expectation}
\begin{align*}
e^{-s\nu}(\E_{0} L^\nu + \E_{0} L^{-\nu}) &\leq  \left(\nu^2 \frac{d\epsilon^4}{2\sigma^4}-s\nu\right)\left( 1  + e^{(3/2)\nu d \epsilon^4/\sigma^4} \right) \\
 &\leq  \exp \left(-\frac{s^2}{2} \frac{\sigma^2}{2\epsilon^2} \right)\left(1+e^{(3/2)s}\right).
\end{align*}
Hence, by plugging in the last two displays into \eqref{eq: UB:help}, and noting that for $0 < s < 1$
\begin{equation*}
 \frac{1}{2}\left(1+e^{(3/2)s}\right)\left[  1 + e^{-s} + e^{s} \right] \leq 12,
\end{equation*}
we arrive at \eqref{to_show_subgaussianity}, for $\beta$ given in \eqref{def:beta:subgauss}, concluding the proof of the lemma.
\end{proof}

\section{Additional technical lemmas}\label{sec: technical:lemmas}
In this subsection, we collect technical lemmas and their proofs.

{

The following lemma is essentially Lemma 11 in \cite{cai_distributed_2020} adopted the setting in this article.
\begin{lemma}[Multivariate Gaussian estimation SPDI]\label{lem : multivar gaussian estimation SPDI}
Let $R$ a $d$-dimensional vector of independent Rademacher variables, $V$ be an independent $\text{Ber}(1/2)$ random variable and let $\Pi$ denote the distribution of $\mu = \epsilon V R$ where $\epsilon > 0 $ is a constant. Suppose that the random vector $X=(X_1,\dots, X_d)$ satisfies $X|\mu \sim N( \mu, \sigma^2 I_d)$ and that $T$ is a discrete random variable such that $V \to \mu \to X \to Y$ forms a Markov chain. Then,
\begin{equation*}
I(\mu; T|V) \leq 128 \left(\frac{\epsilon}{\sigma} \right)^2 I(X;T|V=1).
\end{equation*}
\end{lemma}
\begin{proof}
Write $R_1,\dots,R_d$ for the coordinates of $R$ and write for $k \leq d$, $R_{1:k} := (R_1,\dots,R_k)$ and $X_{1:k} = (X_1,\dots,X_k)$. Conditionally on $V=0$, $\mu = 0$ with probability $1$, so $I(\mu; T | V = 0) = 0$. Conditionally on $V=1$, $\mu = \epsilon R$. Combining these facts with the chain rule for mutual information,
\begin{align*}
I(\mu; T |V) = \frac{1}{2} I(\epsilon R ; T | V=1) &= \frac{1}{2}\underset{k=1}{\overset{d}{\sum}} I( \epsilon R_k ; T | V=1, R_{1:k-1})\\ &= \frac{1}{2}\underset{k=1}{\overset{d}{\sum}} I( \epsilon R_k ; T | V=1),  
\end{align*}
where the last equality follows from the fact that the coordinates of $R$ are independent, so $R_{k+1}$ is independent of $R_{1:k}$. Furthermore, $R_k|V=1 \to X_k |V=1 \to T|V=1$ forms a Markov chain with $(X_k|R_k,V=1) \sim  N(\epsilon R_k,\sigma^2)$. Consequently, by applying Lemma 14 in \cite{cai_distributed_2020},
\begin{equation*}
I( \epsilon R_k ; T | V=1) \leq 64 \left( \frac{2\epsilon}{\sigma} \right)^2 I( X_k ; T | V=1).
\end{equation*}
The proof is now finished by observing that $X_{k+1}|V=1$ is independent of $X_{1:k}|V=1$, so combining the above inequality with the chain rule of mutual information, we obtain
\begin{equation*}
I(\mu; T |V) \leq 128 \left( \frac{\epsilon}{\sigma} \right)^2 \underset{k=1}{\overset{d}{\sum}} I( X_k ; T | V=1, X_{1:k-1})  = 128 \left(\frac{\epsilon}{\sigma} \right)^2 I(X;T|V=1).
\end{equation*}
\end{proof}}

\begin{lemma}\label{lemma_coshZ_expectation}
Let $Z\sim N(0,1)$ and let $\nu \in \R$ such that $|\nu| \epsilon^2/\sigma^2< 1/2$. It holds that
\begin{equation}\label{cosh_exp_bound}
\E  \cosh^\nu\left(\frac{\epsilon}{\sigma} Z\right) \leq \exp \left({\nu \frac{\epsilon^2}{2\sigma^2}+\nu^2 \frac{3\epsilon^4}{2\sigma^4}} - \mathbbm{1}_{\{\nu < 0\}} \frac{3}{2}\nu \frac{\epsilon^4}{\sigma^4} \right).
\end{equation}
\end{lemma}
\begin{proof}
First assume that $\nu \geq 0$. Using $\cosh(x) \leq e^{x^2/2}$ we find
\begin{equation*}
\E  \cosh^\nu\left(\frac{\epsilon}{\sigma} Z\right) \leq \E e^{\nu \frac{\epsilon^2}{2\sigma^2} Z^2}.
\end{equation*}
In view of Lemma \ref{subgaussianity_lemma},
\begin{equation*}
\E e^{\lambda (Z^2-1)} \leq e^{2\lambda^2} \text{ for all } 0 \leq \lambda \leq 1/4.
\end{equation*}
Applying this to the second last display yields \eqref{cosh_exp_bound}. 

Consider now the case that $\nu < 0$. We have 
\begin{align*}
\frac{d}{dx} \cosh^\nu\left(\frac{\epsilon}{\sigma} x\right) &= \nu \frac{\epsilon}{\sigma} \cosh^\nu\left(\frac{\epsilon}{\sigma} x\right) \tanh \left(\frac{\epsilon}{\sigma}x\right),  \\
\frac{d^2}{dx^2} \cosh^\nu\left(\frac{\epsilon}{\sigma} x\right) &= \nu \frac{\epsilon^2}{\sigma^2} \cosh^\nu\left(\frac{\epsilon}{\sigma} x\right) \left[ ( \nu -1 ) \tanh^2 \left(\frac{\epsilon}{\sigma}x\right)  + 1 \right]=:\tau(x)
\end{align*}
Since $\cosh(0)=1$ and $\tanh (0)=0$, a second order Taylor expansion of $x \mapsto \cosh^\nu\left(\frac{\epsilon}{\sigma} x\right)$ about $0$ yields
\begin{align*}
\E  \cosh^\nu\left(\frac{\epsilon}{\sigma} Z\right) &= \E \left[1  + \frac{Z^2}{2!} \tau(r_Z Z)\right], \text{ for some}\quad   r_Z \in [0,1].
\end{align*}
Since $\tanh^2(x) \leq x^2$ and $\cosh(x) \geq 1$ for all $x \in \R$,
\begin{align*}
\E \frac{Z^2}{2!} \tau(r_Z Z) \leq \nu \frac{\epsilon^2}{2\sigma^2} \bigg[  (\nu-1) \frac{\epsilon^2}{\sigma^2}\E  r_Z^2 Z^4    +1\bigg]\leq \nu \frac{\epsilon^2}{2\sigma^2} \bigg[  (\nu-1) \frac{3\epsilon^2}{\sigma^2}  +1\bigg] .
\end{align*}
Then by combining the above two displays
\begin{equation*}
\E  \cosh^\nu\left(\frac{\epsilon}{\sigma} Z\right) \leq \exp \left({\nu \frac{\epsilon^2}{2\sigma^2}+\nu^2 \frac{3\epsilon^4}{2\sigma^4}} - \frac{3}{2}\nu \frac{\epsilon^4}{\sigma^4} \right),
\end{equation*}
\end{proof}

The next lemma gives a well-known sufficient (and necessary) condition for the sub-Gaussian distribution. In the literature we did not find the present, required form of the lemma, hence for completeness we also provide its proof.

\begin{lemma}\label{lem: prob_subg_implies_lapl_subg}
Let $X$ a mean-zero random variable satisfying
\begin{equation*}
\P \left( |X| \geq s \right) \leq C \exp \left( - \frac{s^2}{2 \beta} \right)
\end{equation*}
for some $C\geq 2$, $\beta > 0$ and for all $s \in [0,\infty)$. Then,
\begin{equation*}
\E e^{s X} \leq e^{2\beta C s^2/2}. 
\end{equation*}
\end{lemma}
\begin{proof}
For $k\in \N$, we have
\begin{align*}
\E |X|^k = \int_0^\infty \P\left( |X|^k > t \right) dt 
\leq C \int_0^\infty \exp \left( - \frac{t^{2/k}}{2 \beta} \right) dt. 
\end{align*}
Changing coordinates to $u =t^{2/k}/(2 \beta)$ yields that the right hand side display equals
\begin{align*}
\frac{C}{2} (2\beta)^{k/2} k \int_0^\infty e^{-u} u^{k/2-1} du &=  \frac{C}{2} (2\beta)^{k/2} k \Gamma(k/2).
\end{align*}
By the dominated convergence theorem, $\E X=0$, and $C\geq 2$,
\begin{align*}
\E e^{s X} &= 1 + \underset{k=2}{\overset{\infty}{\sum}} \frac{s^k \E X^k}{k!} 
\leq 1 + \frac{C}{2} \underset{k=2}{\overset{\infty}{\sum}} \frac{(2\beta   s^2)^{k/2}  \Gamma(k/2)}{(k-1)!} \\
&\leq 1 +  \underset{k=1}{\overset{\infty}{\sum}} \left[ \frac{(C\beta   s^2)^{k}  \Gamma(k)}{(2k-1)!} +  \frac{(C\beta s^2)^{k+ 1/2}  \Gamma(k+1/2)}{(2k)!} \right].
\end{align*}
Since $\Gamma(k+1/2) \leq \Gamma (k+1) = k \Gamma(k) = k!$ and $(2k)! \geq 2^k(k!)^2$, the latter is further bounded by
\begin{align*}
1 + \left(1 + \sqrt{{C}\beta  s^2} \right) \underset{k=1}{\overset{\infty}{\sum}} \frac{ (C\beta  s^2/2)^{k}  }{ k!} = e^{{C\beta  s^2/2}} + {\sqrt{C\beta s^2}} (e^{{\beta C s^2/2}}-1).
\end{align*}
Since $(e^x - 1) (e^x - \sqrt{x}) \geq 0$, we obtain that
\begin{equation*}
\E e^{sX} \leq e^{\frac{2C \beta  s^2}{2}}.
\end{equation*}
\end{proof}

The following lemma is a well known result and follows from standard calculus, but we included it as we did not find a stand-alone proof. 

\begin{lemma}\label{subgaussianity_lemma}
Let $Z$ be $N(0,1)$, $0 \leq \lambda \leq 1/4$. Then, 
\begin{equation*}
\E e^{\lambda (Z^2-1)} \leq e^{2\lambda^2}.
\end{equation*}
\end{lemma}
\begin{proof}
Using the change of variables $u = z \sqrt{1- 2\lambda}$,
\begin{align*}
\E e^{\lambda (Z^2-1)} &= \frac{1}{\sqrt{2\pi}} \int e^{\lambda (z^2 - 1)} e^{-\frac{1}{2} z^2} dz  \\
 &= \frac{e^{-\lambda}}{\sqrt{2\pi(1-2\lambda)}} \int  e^{-\frac{1}{2} z^2} dz = \frac{e^{-\lambda}}{\sqrt{(1-2\lambda)}}.
\end{align*}
The MacLaurin series of $-\frac{1}{2} \log (1 - 2\lambda)$ reads
\begin{equation*}
 \frac{1}{2} \underset{k=1}{\overset{\infty}{\sum}} \frac{(2\lambda)^k}{k}, 
\end{equation*}
which yields that the second last display equals
\begin{equation*}
\exp \left(\frac{3}{2} \lambda^2 + \frac{1}{2} \underset{k=3}{\overset{\infty}{\sum}} \frac{(2\lambda)^k}{k} \right).
\end{equation*}
If $\lambda \leq 1/4$,
\begin{equation*}
\underset{k=3}{\overset{\infty}{\sum}} \frac{(2\lambda)^k}{k} \leq \frac{(2\lambda)^3}{1-2\lambda} \leq \lambda^2,
\end{equation*}
from which the result follows.
\end{proof}

The next lemma is a standard bound for the KL-divergence, see for instance Lemma 2.7 of \cite{tsybakov:2009}.
   
\begin{lemma}\label{chi_sq_div_bounds_KL}
Let $P,Q$ probability measures on some measure space such that $Q\ll P$. Then,
\begin{equation*}
D_{KL}(P\|Q) \leq \int \left(\frac{dP}{dQ} - 1 \right)^2 dQ.
\end{equation*}
\end{lemma}

\bibliographystyle{acm}
\bibliography{references}

\end{document}